\documentclass[11pt, oneside]{article}
\usepackage[margin=1in]{geometry}
\usepackage[utf8]{inputenc}
\input{preamble.sty}

\title{Robust Voting Rules from Algorithmic Robust Statistics}
\author{Allen Liu \thanks{Email: \texttt{cliu568@mit.edu}. This work was supported in part by an NSF Graduate Research Fellowship, a Fannie and John Hertz Foundation Fellowship and Ankur Moitra's NSF CAREER Award CCF-1453261 and NSF Large CCF1565235.}\and Ankur Moitra \thanks{Email: \texttt{moitra@mit.edu}. This work was supported in part by a Microsoft Trustworthy AI Grant, NSF CAREER Award CCF-1453261, NSF Large CCF1565235, a David and Lucile Packard Fellowship and an ONR Young Investigator Award.}}
\date{}

\begin{document}

\maketitle
\thispagestyle{empty}

\begin{abstract}

Maximum likelihood estimation furnishes powerful insights into voting theory, and the design of voting rules. However the MLE can usually be badly corrupted by a single outlying sample. This means that a single voter or a group of colluding voters can vote strategically and drastically affect the outcome. Motivated by recent progress in algorithmic robust statistics, we revisit the fundamental problem of estimating the central ranking in a Mallows model, but ask for an estimator that is provably robust, unlike the MLE. 

Our main result is an efficiently computable estimator that achieves nearly optimal robustness guarantees. In particular the robustness guarantees are dimension-independent in the sense that our overall accuracy does not depend on the number of alternatives being ranked. As an immediate consequence, we show that while the landmark Gibbard-Satterthwaite theorem tells us a strong impossiblity result about designing strategy-proof voting rules, there are quantitatively strong ways to protect against large coalitions if we assume that the remaining voters voters are honest and their preferences are sampled from a Mallows model. Our work also makes technical contributions to algorithmic robust statistics by designing new spectral filtering techniques that can exploit the intricate combinatorial dependencies in the Mallows model. 
\end{abstract}

\clearpage
\pagenumbering{arabic} 

\section{Introduction}

\subsection{Background}

Voting rules are general procedures for aggregating preferences over a set of alternatives. They have important applications in elections, collaborative filtering \cite{lu2011learning}, human computation \cite{mao2013better} and crowd-sourcing. In a sense, there are two competing approaches:

\begin{enumerate}
    \item[(1)] The {\em axiomatic} approach formalizes various constraints on how the outcome should vary under changing voter preferences. The goal is to find a voting rule that satisfies one or more such constraints. Unfortunately there are strong impossibility results that say that even some simple and natural axioms are impossible to simultaneously achieve \cite{arrow2012social}. 
    
    \item[(2)] The {\em maximum likelihood estimation} approach posits that there is some absolute sense in which some of the alternatives are better than others. Furthermore each voter's preference is obtained as a noisy estimate of the ground truth ranking. In this context, it is natural to use the maximum likelihood estimator as a voting rule \--- i.e. to select the ranking that has the highest likelihood of generating the observed preferences according to the given noise model. 
    
\end{enumerate}

The MLE approach originated in the work of Condorcet \cite{de2014essai} in the $18$th century. Condorcet studied a model wherein each voter ranks each pair of candidates correctly (i.e. consistently with the ground truth ranking) with some probability $p > \frac{1}{2}$. In the case of two candidates, the MLE can be obtained by a simple majority vote. In the case of three candidates, there can be cycles so that $A$ is preferred to $B$, $B$ is preferred to $C$ and $C$ is preferred to $A$ among a majority of the voters. Condorcet proposed a method for breaking these cycles. Unfortunately these rules are consistent only for the case of three candidates. A solution to the more general model with any number of candidates would have to wait almost two centuries until Young \cite{young1995optimal} showed that a correct application of the maximum likelihood principle leads to selecting an ordering that minimizes the number of disagreements with the voters, and this is exactly the distance that is minimized by the Kemeny rule. 
Many recent works have focused on characterizing and computing the MLE in various ranking models \cite{braverman2009sorting, soufiani2012random, soufiani2013preference} and investigating its other statistical properties \cite{xia2018bayesian, noothigattu2020axioms}. Much of this work has been motivated by computing-enabled applications where information about user preferences is abundant and there are sometimes huge numbers of alternatives that must be ranked \cite{lu2011learning, mao2013better}. 

An important work of Conitzer and Sandholm \cite{ConitzerS05} flipped the question around: {\em Which voting rules can be expressed as the MLE in some noise model?} This viewpoint furnishes new insights about what sorts of assumptions we are making about how preferences are generated, when we adopt some particular voting rule. And when a voting rule can only be realized as the MLE on some contrived noise model, it opens the door to the possibility that there might be a modification to the voting rule that permits the noise model to be more justifiable.  Other works have applied the MLE framework to related problems, such as aggregating partial preferences \cite{conitzer2009preference} and selecting a group of alternatives \cite{ProcacciaRS12}. There has also been work on discovering equivalences between the MLE approach and distance-based rules \cite{elkind2010good}. And still other works have investigated its sample complexity \cite{braverman2009sorting, caragiannis2013noisy}. 

Despite the success of the MLE paradigm, and the fact that many popular voting rules can be cast in this framework, there is an important shortcoming when it comes to robustness: The MLE can usually be badly corrupted by a single outlying sample. Thus if we are to put stock in the underlying noise model, we face the challenge that there is nothing preventing a single voter or a small group of colluding voters from voting strategically in a way that they drastically affect the outcome. 
Recently there has been considerable progress in algorithmic robust statistics \cite{diakonikolas2019robust, lai2016agnostic, diakonikolas2017being}. In many cases, these works provide estimators that are accurate, efficiently computable and, unlike the MLE, are provably robust to a constant fraction of their data being arbitrarily corrupted. Thus we ask if this toolkit can provide more robust alternatives to the MLE in the context of voting theory: 

\begin{quote}
    {\em Are there voting rules that accurately estimate the ground truth ranking, can be efficiently computed, but whose outcome cannot be changed by much even by a large group of colluding voters?}
\end{quote}

\noindent Many works in voting theory have investigated robustness, albeit in different ways. Caragiannis et al. \cite{caragiannis2013noisy} gave a characterization of the distances that lead to voting rules that are robust against monotone changes. Later Caragiannis et al. \cite{caragiannis2014modal} explored the prospect of designing voting rules that recover the ground truth against a family of noise models. Procaccia et al. \cite{procaccia2016voting} studied a model-free setting, but where all preferences are assumed to be close to the central ranking in some metric. And Benade et al. \cite{benade2017making} worked with the relaxed assumption that preferences are close on average. These works can be thought of as exploring some particular directions along which the noise model might be misspecified. Instead, we 
will allow truly adversarial and unbounded corruptions. By virtue of its generality, our model will allow us to give provable guarantees in a distributional sense on how much any voter or group of colluding voters can influence the outcome by voting strategically (see Section~\ref{sec:interp}).


\subsection{Our Results}

In this work, our main goal is to design a provably robust and computationally efficient algorithm for estimating the central ranking in the Mallows model. This estimator can be thought of as a robust alternative to the MLE. The Mallows model is defined as follows:

\begin{definition}[Kendall-Tau Distance]
For two permutations $\pi, \sigma$ on $[n]$, the Kendall-tau distance $d_{\KT}(\pi, \sigma)$ is defined to be the number of relative inversions between $\pi$ and $\sigma$.
\end{definition}

\begin{definition}[Mallows Model]\label{def:mallows}
A Mallows model on $n$ elements is described by two parameters: a real number  $0 \leq \phi \leq  1$ and a base permutation (a.k.a. a central ranking) $\pi^*$ on $[n]$.  Then a Mallows model defines a distribution on permutations of $[n]$ where the probability of $\pi$ is 
\[
\frac{\phi^{d_{\KT}(\pi, \pi^*)}}{Z_n(\phi)}
\]
where $Z_n(\phi) = \sum_{\pi} \phi^{d_{\KT}(\pi, \pi^*)}$ is the normalizing constant.  We denote a Mallows model using $M(\phi, \pi^*,n)$ or $M(\phi, \pi^*)$ when the number of elements is clear from context.
\end{definition}

\begin{remark}
It can be shown that $Z_n(\phi)$ only depends on $n$ and $\phi$ but not on $\pi^*$.
\end{remark}

This important model was introduced by Colin L. Mallows in 1957 \cite{MallowsReconstruction} and has been a mainstay in social choice theory ever since. It has found wide-ranging applications, and has been refined and built upon over the years \cite{fligner1986distance,doignon2004repeated,marden1996analyzing, mukherjee2016estimation, busa2019optimal}. In general it is hard to compute the MLE. But when the samples actually come from the Mallows model, there are efficient algorithms that succeed with high probability \cite{MallowsReconstruction}. Instead of getting samples directly from a Mallows model, we will work in the strong contamination model:

\begin{definition}[Strong Contamination Model]\label{def:corruption}
We say that a set of samples $\gamma_1, \dots , \gamma_s$ is an $\eps$-corrupted sample from a distribution $\mcl{D}$ if it is generated as follows.  First $\theta_1, \dots , \theta_s$ are sampled i.i.d. from $\mcl{D}$.  Then a (malicious, computationally unbounded) adversary observes $\theta_1, \dots , \theta_s$ and replaces up to $\eps s$ of them with any samples it chooses.  The adversary may then reorder the vectors arbitrarily and output them as $\gamma_1, \dots , \gamma_s$.
\end{definition}

\noindent In the context of voting, this means that an $\epsilon$-fraction of the voters are allowed to submit nonsensical or even strategically manipulated preferences. Even worse, they are allowed to inspect the preferences of all the honest voters and then decide, in a coordinated manner and with unbounded computational resources, on what preferences to submit. We discuss this further in Section~\ref{sec:interp}.

In the strong contamination model, it is not always possible to learn the central ranking exactly. In particular, if there are two Mallows models that are $\epsilon$-close in total variation distance, the adversary can change one distribution into another so that it is impossible to determine which among them was the original model. Still, this begs the question(s): {\em Is it possible to approximately determine the central ranking? And if so, what is the appropriate distance to use to measure closeness? }
Our first main result is a tight characterization of the total variation distance between two Mallows models in terms of a distance between their central rankings.  It turns out that the right distance is the $L_2$ distance between their position vectors. 
\begin{definition}[Position Vector]\label{def:position}
For a permutation $\pi$ on $[n]$, we define its position vector as the $n$-dimensional vector  
\[
v_{\pi} = ( \pi^{-1}(1) , \dots ,  \pi^{-1}(n) )  \,.
\]
\end{definition}

\begin{theorem}[Informal, see Theorem \ref{thm:TV-characterization} and Theorem \ref{thm:TV-lowerbound}]\label{thm:TV-informal}
For a scaling parameter $\phi$ and two permutations $\pi^*, \sigma^*$ on $[n]$, we have
\[
D_{\TV}( M(\phi, \pi^*),  M(\phi, \sigma^*) )  	\asymp (1 - \phi) \norm{v_{\pi^*} - v_{\sigma^*}}
\]
as long as $\norm{v_{\pi^*} - v_{\sigma^*}} \leq 1/(1 - \phi)$ where the $	\asymp$ denotes that the two sides are within a constant factor.  Otherwise $D_{\TV}( M(\phi, \pi^*),  M(\phi, \sigma^*) ) = \Omega(1)$. 
\end{theorem}

In algorithmic robust statistics, the goal is to estimate the distribution in total variation distance. However TV is usually cumbersome to work with and so the first step is to find a suitable parameter distance that is equivalent to TV distance. For many of the early successes of algorithmic robust statistics \cite{diakonikolas2019robust, lai2016agnostic}, such tight characterizations of the TV distance were immediate. Surprisingly, despite the fact that the Mallows model is generally well-understood, such a characterization appears to have been unknown prior to our work. 

The $L_2$ distance is also natural and intuitive in its own right. In words, it measures the $L_2$ norm of the vector of distances that each element moves when going from $\pi^*$ to $\sigma^*$.
 In Section~\ref{sec:whyL2} we show that it is the strongest distance we can use in the sense that for any other distance on permutations where we can estimate the central ranking well, we can obtain essentially the same guarantees by estimating in $L_2$ and using the relationship between the two distances instead.

Now that we have a statistical characterization of what is possible, the natural question is: 

\begin{quote}
    {\em Are there computationally efficient algorithms for robustly learning a Mallows model?}
\end{quote}

\noindent  The main emphasis in algorithmic robust statistics is on getting dimension independent guarantees, which in our context would mean that we want to estimate the central ranking in $L_2$ within some error that is only a function of the corruption level $\epsilon$. The key point is that the error should not grow with the number of alternatives. Furthermore since robustly learning a Mallows model is a basic and fundamental problem, we could even hope for $\widetilde{O}(\epsilon)$ accuracy, which would be nearly optimal. This is our main result:

\begin{theorem}[Main]\label{thm:main}
Assume $\eps < 0.1$.  Given $\eps$-corrupted samples $\gamma_1, \dots , \gamma_s$ from an unknown Mallows model $M(\phi, \pi^*)$ on $n$ elements where $s = (n/\eps)^2 \poly(\log (n/\eps))$, there is an algorithm that runs in $\poly(s)$ time, and with $1 - (\eps/n)^8$ probability, outputs a Mallows model $M(\wt{\phi}, \wt{\pi^*})$ such that
\begin{itemize}
    \item[(a)] $ \norm{v_{\pi^*} - v_{\wt{\pi^*}}} \leq O\left( \frac{\eps \log 1/\eps}{ 1 - \phi} \right)$ and $| \phi - \wt{\phi}| \leq O(\eps/\sqrt{n})$
    \item[(b)] $D_{\TV}\left(M(\phi, \pi^*) ,  M(\wt{\phi}, \wt{\pi^*})\right) \leq O(\eps \log (1/\eps))$
\end{itemize}
\end{theorem}


The guarantees in Theorem \ref{thm:main} about estimating $\pi^*$ and the distribution $M(\phi, \pi^*)$ are optimal up to a logarithmic factors. The basic idea is to robustly estimate the position vector that encodes the position of each alternative in the central ranking. When the covariance is known and the distribution is subgaussian, there are already strong guarantees that are known \cite{diakonikolas2017being}. However the results become quantitatively weaker (and this is information-theoretically necessary in general) when we only have upper bounds on the covariance matrix. 

The main challenge is that in a Mallows model it is hard to get a handle on the covariance. The entries \--- i.e. which alternatives appear in which position in a sample \--- exhibit complex combinatorial dependencies. There are natural algorithms \cite{doignon2004repeated} for sampling from a Mallows that give us some insight. However there is a chicken-and-egg problem: If we know the central ranking, we can think about a sample from the Mallows model as being generated by a sequence of independent decisions about where to insert an element into the ranking thus far. But this is only true if we are adding elements in the order in which they appear in the central ranking. Thus if we knew the central ranking, we could get better control on the covariance, but we are trying to get better control on the covariance precisely so that we can better estimate the central ranking. Ultimately we give an iterative algorithm that uses our estimate of the central ranking to refine our estimate of the covariance, and vice-versa. For a more detailed overview, see Section~\ref{sec:tech}. 


\subsection{Interpretation of Results}\label{sec:interp}

Here we return to voting theory and explain some of the consequences of our main results, as well as provide further discussion. 

\subsubsection{Trivial Regimes}

There are some settings, depending on the parameter $\phi$, where it is trivial to robustly learn the central ranking even in the presence of corruptions. If $1 - \phi > \eps$ then without corruptions most of the samples would have $i$ and $j$ in the correct order. Thus even with a small but constant fraction of corruptions we could learn the relative order of $i$ and $j$ in the central ranking just by taking a majority vote. However this is a less interesting setting because in a Mallows model, in a typical sample, each element is displaced from its location in the central ranking by at about $1/(1 - \phi)$ positions. Thus assuming $1 - \phi > \eps$ severely limits its applicability. In contrast, our algorithm makes no assumptions about $\phi$ and continues to succeed even if elements are typically quite far from their location in the central ranking.

\subsubsection{Distributional Group Strategy-proofness}

We say that a voting rule is {\em strategy-proof} if a voter can never do strictly better by lying about his preferences. This is a strong way to formulate what it means to be immune to strategic behavior. When there are just two candidates, the majority voting rule is strategy-proof. However when there are three or more candidates, there are no non-trivial strategy-proof voting rules. In particular, we say that a voting rule is a {\em dictatorship} if there is a voter who always gets his preferred outcome, regardless of the preferences of the other voters. And we say a voting rule is {\em onto} if it is possible for every candidate to win. The landmark Gibbard-Satterthwaite theorem tells us:

\begin{theorem}\cite{gibbard1973manipulation, satterthwaite1975strategy}
If there are three or more candidates, any strategy-proof voting rule that is onto is a dictator.
\end{theorem}

We may seem to be at an impasse. While strategy-proofness is an important property to aim for, it is out of reach in a strong sense. Nevertheless it is possible to salvage various approximate notions of strategy-proofness that give conditions under which the outcome cannot usually be changed by small coalitions of voters. See Mossel et al. \cite{mossel2013smooth} and references therein. In our setting of interest, we assume that a majority of the voters are honest and their preferences are sampled from a Mallows model but the remainder are allowed to collude and lie about their preferences. We call this {\em distributional group strategy-proofness} since it crucially relies on having distributional assumptions on the honest voters. Now as a corollary of Theorem~\ref{thm:main}, we get the following guarantee:

\begin{corollary}\label{thm:voting}
Assume that there are $m$ voters and all but $c$ of them vote truthfully with their preferences drawn from an unknown Mallows model $M(\phi, \pi^*)$.  Also assume $c \leq 0.1m$.  Even if the remaining $c$ voters collude in an arbitrary fashion, the voting rule in Theorem~\ref{thm:main} satisfies the property
\[
\norm{v_{\wt{\pi^*}} - v_{\pi^*}} \leq\left(   O\left(\frac{c}{m} \log \frac{m}{c}  \right)  + \left(\frac{\wt{O}(n \sqrt{m})}{m} \right) \right) \cdot \frac{1}{ 1 - \phi}  \,.
\]
\end{corollary}

As mentioned above, Mossel et al. \cite{mossel2013smooth} also study the power of coalitions in influencing the outcome of an election under distributional assumptions. We make stronger assumptions about the distribution, since the preferences of the honest voters are drawn from a Mallows model, but we obtain quantitatively sharper results: $(1)$ We allow coalitions of linear size, whereas Mossel et al. \cite{mossel2013smooth} can only handle coalitions of size at most $\sqrt{m}$. $(2)$ In their setting, the number of alternatives is a constant. For us, not only can the number of alternatives be large, we moreover want algorithms whose running time and sample complexity are polynomial in $n$ and crucially we get robustness guarantees that are independent of $n$. $(3)$ Our voting rule approximately recovers the entire ranking, whereas Mossel et al. consider the probability that a coalition can manipulate the top ranking candidate. These guarantees are incomparable, in particular because Mossel et al. need to assume a margin lower bound. See also Xia \cite{xia2012many} who defines a class of strategic behaviors that captures many common types of manipulation and studies the power of coalitions of varying size. 

To see why the above corollary from Theorem~\ref{thm:main}, we can view the samples as $c/m$-corrupted samples from a Mallows model.  Recall that the guarantees of Theorem~\ref{thm:main} hold as long as the number of samples is at least $(n/\eps)^2 \poly(\log (n/\eps))$.  Thus, if $c \geq \wt{O}(n\sqrt{m})$, then we can apply the theorem (in this case, the first term upper bounds the LHS).  Otherwise, we can simply pretend that there are additional corruptions so that the number of corruptions is $\wt{O}(n \sqrt{m})$ (in this case, the second term upper bounds the LHS).

In the bounds above, the second term corresponds to sampling noise, which is unavoidable even for non-robust algorithms . The first term bounds how much more the colluding voters can affect the outcome, and the bound of $\wt{O}(c/m)$ is nearly optimal and independent of the number of alternatives being ranked. Thus our main results show that while strategy-proofness is too strong to hope for, there are relaxed notions that can be achieved. The guarantees above seem particularly relevant in settings where each voter has to expend effort to cheat (or to determine how to cheat) because the benefit of doing so degrades essentially inversely with the number of voters. More broadly, other recent works have investigated further applications of robust learning in various game theoretic settings, including estimating demand curves 
\cite{chen2020efficiently}, learning optimal auctions \cite{guo2021robust} and contextual search \cite{krishnamurthy2021contextual}. There seem to be many interesting directions left to explore.

\subsection{Related Work}

Braverman and Mossel \cite{braverman2009sorting} gave an algorithm for learning a Mallows model in the noiseless setting. Caragiannis et al. \cite{caragiannis2013noisy} gave optimal bounds on the sample complexity. And Busa-fekete et al. \cite{busafekete2019optimal} gave optimal sample complexity bounds for the more general Mallows block model. In some information-aggregation settings the underlying population is heterogenous \cite{Gormley} and it is natural their preferences as coming from a mixture of two or more Mallows models. Awasthi et al. \cite{awasthi2014learning} gave an algorithm for mixtures of two Mallows models. Liu and Moitra \cite{liu2018efficiently} gave an algorithm for the general case and improved sample complexity bounds were given in \cite{mao2020learning}. There has also been work either on learning or showing identifiability in other ranking models, like the Plackett-Luce model \cite{Plackett} and the Cayley-Mallows model \cite{de2018learning}. {\em An important open question from our work is to give efficient algorithms for robustly learning other ranking models. }

Our work builds on a long line of progress in high-dimensional algorithmic robust statistics, which includes algorithms for robustly estimating a single Gaussian \cite{diakonikolas2019robust, lai2016agnostic, diakonikolas2017being}, robust linear regression \cite{klivans2018efficient, diakonikolas2019sever, bakshi2020robust, chen2020online} and mixture models \cite{hopkins2018mixture, liu2021settling, bakshi2020robustly, liu2021learning}. In particular there is by now an understanding of some of the overarching techniques, including spectral filtering, and how to apply them in various settings. However there has been considerably less work on robust algorithms for unsupervised learning of discrete distributions. Diakonikolas et al. \cite{diakonikolas2019robust} give robust algorithms for learning product distributions on the hypercube. Crucially, the coordinates are independent. They also give extensions to mixtures of two product distributions. In our setting, our samples are permutations and a major challenge is exactly to tame the dependencies between the coordinates (i.e. the locations in which different elements land in the preference order). 

There is also a rich literature studying the manipulability of elections. The landmark results of Gibbard and Satterthwaite \cite{gibbard1973manipulation, satterthwaite1975strategy} establish that any voting rule with at least three candidates that is not a dictator is manipulable in the sense that there are situations where voters would do strictly better to vote strategically rather than truthfully. There are also quantitative versions \cite{friedgut2011quantitative, mossel2015quantitative}. Our results are in a rather different setting in that they show there are approximately manipulation-proof voting rules provided that a majority of the voters are honest and their preferences come from a Mallows model.

\section{Technical Overview}\label{sec:tech}

In this section, we give an overview of our techniques.

\subsection{Review: Sampling from a Mallows Model via Insertion Procedure}

We begin by reviewing a standard, but extremely useful fact.  From the definition of a Mallows model in Definition \ref{def:mallows}, it is not immediately clear how to efficiently draw a sample.  However, there is a standard procedure (see e.g. \cite{doignon2004repeated}) for sampling by inserting the elements one by one.  We call this the insertion procedure.

\begin{lemma}[Sampling via Insertion \cite{doignon2004repeated}]\label{lem:insertion-sampling}
Let $M(\phi, \pi^*)$ be a Mallows model on $n$ elements.  A sample $\pi$ from  $M(\phi, \pi^*)$ can be drawn using the following recursive process for $a = 1,2, \dots $.  Assuming that we have sampled a permutation on $ \{ \pi^*(1), \dots , \pi^*(a - 1) \}$ so far, we insert the element $\pi^*(a)$ into each of the $a$ positions with probabilities 
\[
\frac{\phi^{a-1}}{1 + \phi + \dots + \phi^{a-1}}, \frac{\phi^{a-2}}{1 + \phi + \dots + \phi^{a-1}}, \dots , \frac{1}{1 + \phi + \dots + \phi^{a-1}}
\]
respectively.  We then recurse by inserting the next element $\pi^*(a+1)$ into one of $a+1$ possible positions and so on.
\end{lemma}
Roughly, the insertion procedure involves iteratively inserting each element into a permutation (that starts out empty) according to the insertion distribution specified.

\subsection{Identifying the Right Parameter Distance}

Our first main result is a tight characterization of the TV distance in terms of the $L_2$ distance between position vectors (Definition~\ref{def:position}).  


\paragraph{Upper Bounds on TV} The first part is straightforward:  In order to prove an upper bound on the TV distance, we estimate the KL divergence between $M(\phi, \pi^*)$ and  $M(\phi, \sigma^*)$ and then apply Pinsker's inequality. This is analytically tractable because for any permutation $\pi$, the log likelihood ratio
\[
X = \log \frac{\Pr_{M(\phi, \pi^*)}[\pi]}{\Pr_{M(\phi, \sigma^*)}[\pi]}
\]
can be decomposed combinatorially as a sum over pairs of elements whose order is different in $\pi^*$ and $\sigma^*$.  The remainder is a direct calculation.  See the proof of Theorem \ref{thm:TV-characterization} for more details. 

\paragraph{Lower Bounds on TV} It is more challenging to lower bound the TV distance. Recall that $X$ is the random variable denoting likelihood ratio. First we show a strong concentration inequality for $X$, namely that it has subgaussian tails. We then use this property to establish a reverse Pinsker's inequality.  To prove concentration for $X$, we can imagine sampling a permutation $\pi \sim M(\phi, \pi^*)$ using the insertion procedure.  Next we use martingale concentration bounds to show that the distribution of $X$ concentrates.  Formally, consider the Doob martingale $X_0, \dots , X_n$ where $X_t$ is the expected value of $X$ after inserting the first $t$ elements.  To obtain concentration for $X_n = X$, by Azuma's inequality, it suffices to bound $\sum_t \max|X_t - X_{t-1}|^2$.  It turns out that we are able to bound each of the differences individually through a detailed analysis of the insertion procedure and a characterization of how each insertion contributes to the probabilities of various pairs of elements being inverted.  See the proof of Theorem \ref{thm:TV-lowerbound} for more details.

\subsection{Spectral Filtering with Combinatorial Dependencies}
Next we explain the key ideas in our main algorithm and its analysis.  In light of Theorem \ref{thm:TV-informal}, it suffices to estimate the position vector $v_{\pi^*}$ to accuracy $\wt{O}(\eps/(1 - \phi))$.  We first note that this is equivalent up to constant factors to estimating the mean of the distribution $v_{\pi}$ for $\pi \sim M(\phi, \pi^*)$ and then sorting the entries of $v_{\pi}$ (see Claim~\ref{claim:sorting}).  We have thus reduced our problem to a robust mean estimation problem, but one where the covariance is driven by the Mallows model. Our algorithm works in two phases.  First, we obtain a rough estimate $\wh{\pi}$ for the true permutation with accuracy $O(\sqrt{\eps})$.  We then iteratively refine this estimate.

\paragraph{Iterative Refinement} 
We focus on the iterative refinement step as it is the main novelty in our algorithm.  To simplify the discussion and notation, we will consider the case $\pi^* = \id$, where $\id$ is the identity permutation.  Robust mean estimation hinges on understanding the covariance of the underlying distribution \cite{diakonikolas2019recent}.  The covariance of the distribution of $v_{\pi}$ for $\pi \sim M(\phi, \id)$ is complicated by the fact that there are nontrivial dependencies between the positions of various elements.  However, we can remove these dependencies with the following decomposition.  We can write
\[
v_{\pi} = v_{\id} - v_{\pi}^{\leftarrow } + v_{\pi}^{\rightarrow}
\]
where $v_{\pi}^{\leftarrow}$ is a vector whose $i$th entry is the number of elements $j < i$ that occur after $i$ in $\pi$ and $v_{\pi}^{\rightarrow}$ is a vector whose $i$th entry is the number of elements $j > i$ that occur before $i$ in $\pi$.  It now suffices to understand the distributions of $v_{\pi}^{\leftarrow}$ and $v_{\pi}^{\rightarrow}$ separately.  The key is that the insertion procedure implies that the entries of $v_{\pi}^{\leftarrow}$  are independent and thus its covariance matrix is diagonal and can be explicitly computed. The same is true for the entries of $v_{\pi}^{\rightarrow}$.  We get that, aside from the first few and last few entries, which will behave differently, the variances are all close to $\sqrt{\phi}/(1 - \phi)$.  Thus, we have
\[
\Cov_{\pi \sim M(\phi, \id)}\left( \frac{1 - \phi}{\sqrt{\phi}} \trunc(v_{\pi}^{\leftarrow} )\right) \sim I
\]
where $\trunc$ denotes deleting the first few and last few entries of the vector.  It turns out that there is a simple reduction showing that it suffices to only work with truncated vectors (see Section~\ref{sec:padding}).

Still, we cannot use any of the above in our algorithm because computing the vector $v_{\pi}^{\leftarrow }$ relies on knowing the hidden permutation which we do not actually know.  To circumvent this, we use our rough estimate $\wh{\pi}$ instead.   Define $v_{\pi}^{\leftarrow \wh{\pi}}$ to be a vector whose $i$th entry is the number of elements $j$ with $\wh{\pi}^{-1}(j) < \wh{\pi}^{-1}(i)$  and $\pi^{-1}(i) > \pi^{-1}(j)$ and define $v_{\pi}^{\rightarrow \wh{\pi}}$ to be a vector whose $i$th entry is the number of elements $j$ with $\wh{\pi}^{-1}(j) > \wh{\pi}^{-1}(i)$  and $\pi^{-1}(i) < \pi^{-1}(j)$.  We still have the decomposition 
\[
v_{\pi} = v_{\wh{\pi}} - v_{\pi}^{\leftarrow \wh{\pi} } + v_{\pi}^{\rightarrow \wh{\pi}}
\]
but now all of the above are quantities that we can compute.  Of course this new decomposition has the issue that the entries of $v_{\pi}^{\leftarrow \wh{\pi} }$ are no longer independent.  However, we will prove that the entries have bounded dependence and control this dependence in terms of $\norm{v_{\wh{\pi}} - v_{\id}}$. The particular quantitative bound we show will allow us to bootstrap our rough estimate in the sense that we can an estimate of $v_{\id}$ to get better bounds on the covariance, which, in turn, will let us get an even better estimate of $v_{\id}$, and so on.  The key lemma is:
\begin{lemma}\label{lem:error-informal}[Informal, see Lemma~\ref{lem:error-in-cov}]
Let $\wh{\pi}$ be a permutation that is sufficiently close to $\id$.  Then we have
\[
\norm{\Cov_{\pi \sim M(\phi, \id)}\left( \frac{1 - \phi}{\sqrt{\phi}} \trunc\left(v_{\pi}^{\leftarrow \wh{\pi}} \right) \right) - I}_{\op} \leq  O(\norm{v_{\wh{\pi}} - v_{\id}} ( 1 - \phi) ) \,.
\]
and the same holds for $v_{\pi}^{\rightarrow \wh{\pi}}$.
\end{lemma}

We now sketch the proof of the above lemma.  By the insertion procedure, the covariance of  $v_{\pi}^{\leftarrow \wh{\pi}}$ for $\pi \sim M(\phi, \wh{\pi})$ is diagonal.  Now, we reason about the covariance of $v_{\pi}^{\leftarrow \wh{\pi}}$ for $\pi \sim M(\phi, \id)$ by reasoning about the difference between sampling from $M(\phi, \wh{\pi})$ and $M(\phi, \id)$.  In particular for any $i$, we can write the difference
\[
v_{\pi}^{\leftarrow}[i] - v_{\pi}^{\leftarrow \wh{\pi}}[i] = \sum_{\substack{j \text{ s.t. } j < i \\  \wh{\pi}^{-1}(j) > \wh{\pi}^{-1}(i)}} 1_{\pi^{-1}(j) > \pi^{-1}(i)} - \sum_{\substack{ j \text{ s.t. } j > i \\  \wh{\pi}^{-1}(j) < \wh{\pi}^{-1}(i)}} 1_{\pi^{-1}(j) > \pi^{-1}(i)} \,.
\]
Note that all pairs $i,j$ that appear in the above sum must be inverted in $\wh{\pi}$ (compared to in $\id$).  In particular, these pairs cannot be too far apart if $\wh{\pi}$ is close to $\id$.  Thus, we can argue that all nonzero entries in the covariance of  $v_{\pi}^{\leftarrow \wh{\pi}}$ must actually be within a band of width $\sim \norm{v_{\wh{\pi}} - v_{\id} }$ around the diagonal.  We then bound the sizes of these nonzero entries by analyzing the covariance of the indicator variables i.e. expressions of the form
\[
\Cov_{\pi \sim M(\phi, \id)} \left( 1_{\pi^{-1}(w) > \pi^{-1}(y)}, 1_{\pi^{-1}(x) > \pi^{-1}(z)}\right) 
\]
for some $w,x,y,z$.  The quantitative bounds we need are given in Lemma \ref{lem:indicator-covariance} and Lemma \ref{lem:estimated-cov}.  The proofs of these bounds are involved because we end up needing to quantify the dependencies between the locations of elements in a Mallows model.  While one would expect explicit expressions for such quantities to be unwieldy, we are able to obtain sufficiently strong bounds (that are tight up to constants) through careful combinatorial arguments.

\paragraph{Putting Everything Together}
We now summarize our full algorithm.  We first obtain our rough estimate by arguing that the covariance of the distribution of $v_{\pi}$ for $\pi \sim M(\phi, \pi^*)$ is bounded above and then appealing to existing results in \cite{diakonikolas2017being, diakonikolas2019robust} to solve robust mean estimation up to accuracy $O(\sqrt{\eps})$.  Then we iteratively refine the estimate $\wh{\pi}$ by using stronger robust mean estimation algorithms \cite{diakonikolas2019recent, diakonikolas2020outlier} that give us accuracy guarantees that depend on how well we know the covariance of the underlying distribution.  In particular, Lemma~\ref{lem:error-informal} allows us to argue that we can reduce the error of our estimate $\wh{\pi}$ by a constant factor in each step until we get to $\wt{O}(\eps)$.

\subsection{Paper Organization}
In Section~\ref{sec:prelims}, we define notation and go over some basic properties of Mallows models.  In Section~\ref{sec:position-to-TV}, we show that the TV distance between two Mallows models is equivalent to the $L^2$ distance between their base permutations up to constant factors.  In Section~\ref{sec:tail-bounds}, we prove tail bounds on the distribution generated by a Mallows model that will be used in our robust learning algorithm.  Finally, in Section~\ref{sec:main-alg}, we present our learning algorithm.  We first discuss the generic robust mean estimation subroutines in Section~\ref{sec:robust-mean-est}.  We then show how to get a rough estimate in Section~\ref{sec:rough-estimate} and how to do the iterative refinement in Section~\ref{sec:iterative-refinement}.  Finally, in Section~\ref{sec:hypothesis-test}, we do  post-processing steps to estimate the scaling parameter $\phi$ (in addition to the base permutation).







\section{Basic Properties of Mallows Models}\label{sec:prelims}

In this section, we introduce notation and prove some basic properties about Mallows models that will be used repeatedly later on.

\subsection{Notation}

\begin{definition}
We use $\id_n$ to denote the identity permutation on $[n]$.  When $n$ is clear from context, we may simply write $\id$.
\end{definition}

\begin{definition}
For a vector $v$ with $n$ entries and an element $i \in [n]$, we use $v[i]$ to denote the $i$\ts{th} entry of $v$.
\end{definition}

\begin{definition}[Restriction of Permutation]\label{def:permutation-restriction}
For a permutation $\pi$ on $[n]$ and subset $S \subset [n]$, we use $\pi|_S$ to denote the permutation on the elements of $S$ induced by $\pi$.
\end{definition}

The next two definitions make it convenient to work with the restriction of a permutation $\pi$ to some prefix or suffix of a different permutation $\wh{\pi}$.   

\begin{definition}\label{def:front}
For permutations $\pi, \wh{\pi}$ on $[n]$, define $\pi_{\wh{\pi}:a}$ to be the permutation on the elements $(\wh{\pi}(1), \dots , \wh{\pi}(a))$ induced by $\pi$.
\end{definition}
\begin{definition}\label{def:back}
For permutations $\pi, \wh{\pi}$ on $[n]$, define $\pi_{\wh{\pi}:-a}$ to be the permutation on the elements $(\wh{\pi}(n-a + 1), \dots , \wh{\pi}(n))$ induced by $\pi$.
\end{definition}
\begin{remark}
Note that for any permutations $\pi, \wh{\pi}$ on $[n]$, we have $\pi_{\wh{\pi}:-n} = \pi_{\wh{\pi}:n} = \pi$.
\end{remark}


\subsection{Basic Computations} 

In light of Lemma \ref{lem:insertion-sampling}, it will be convenient to make the following definition.
\begin{definition}\label{def:insertion-distribution}
For an integer $i$ and parameter $0 < \phi < 1$, let $\mcl{D}_{i, \phi}$ be the distribution on integers $\{0,1, \dots , i - 1 \}$ where the probability of obtaining $j$  is exactly 
\[
\frac{\phi^{ j }}{1 + \phi + \dots + \phi^{i - 1}} \,.
\]
We use $\mcl{D}_{\infty, \phi}$ to denote the limiting distribution as $i \rightarrow \infty$ i.e. the probability of obtaining any nonnegative integer $j$ is exactly
\[
(1 - \phi)\phi^{j} \,.
\]
We will drop the subscript $\phi$ and just write $\mcl{D}_i$ or $\mcl{D}_{\infty}$ when it is clear from context.
\end{definition}

It will be useful to note the following explicit expressions for the first two moments of $\mcl{D}_{\infty, \phi}$.
\begin{claim}\label{claim:explicit-means}
We have for all $i$
\[
\E_{x \sim \mcl{D}_{i, \phi}}[ x] \leq \E_{x \sim \mcl{D}_{\infty, \phi}}[ x] =  \frac{\phi}{1 - \phi}  \,.
\]
\end{claim}
\begin{proof}
Note that from the definition, it is clear that $\E_{x \sim \mcl{D}_{i+1}}[x] \geq \E_{x \sim \mcl{D}_{i}}[x]$.  Thus, it suffices to bound the mean of $\mcl{D}_{\infty}$.  This is 
\[
\E_{x \sim \mcl{D}_{\infty}}[ x] = \sum_{j = 0}^{\infty} (1 - \phi)j \phi^j = \frac{\phi}{1 - \phi}
\]
and we are done.
\end{proof}

\begin{claim}\label{claim:explicit-variance}
We have
\[
\Var_{x \sim \mcl{D}_{\infty, \phi}}(x) = \frac{\phi}{(1 - \phi)^2}\,.
\]
\end{claim}
\begin{proof}
Note that 
\[
\E_{x \sim \mcl{D}_{\infty}}[ x^2] = \sum_{j = 0}^{\infty} (1 - \phi)j^2 \phi^j = \sum_{j = 1}^{\infty} (2j - 1) \phi^j = \frac{2\phi}{(1 - \phi)^2} - \frac{\phi}{1 - \phi} = \frac{\phi + \phi^2 }{(1 - \phi)^2} \,.
\]
Thus, by Claim \ref{claim:explicit-means},
\[
\Var_{x \sim \mcl{D}_{\infty}}(x) = \frac{\phi + \phi^2 }{(1 - \phi)^2} - \left(\frac{\phi}{1-  \phi} \right)^2 = \frac{\phi}{(1 - \phi)^2} \,.
\]
\end{proof}

Note that as $i \rightarrow \infty$, the distribution $\mcl{D}_i$ approaches the distribution $\mcl{D}_{\infty}$.  Thus, it is natural to expect that the mean and variance of $\mcl{D}_i$ converge to the mean and variance of $\mcl{D}_{\infty}$.  In fact, since the tails of the distributions decay exponentially, this convergence happens very fast, as formalized in the next claim.

\begin{claim}\label{claim:almost-exact-expressions}
For any parameter $0 < \delta < 1$ and integer $i$, if $i \geq 10^2 \log( 1/\delta)/ (1-\phi)$ then 
\begin{align*}
\left \lvert \E_{x \sim \mcl{D}_{i,\phi}}[x] - \E_{x \sim \mcl{D}_{\infty, \phi}}[x] \right \rvert &\leq \frac{\delta}{(1 - \phi)} \\ 
\left \lvert \Var_{x \sim \mcl{D}_{i, \phi}}[x] - \Var_{x \sim \mcl{D}_{\infty, \phi}}[x] \right \rvert &\leq \frac{\delta}{(1 - \phi)^2}  \,.
\end{align*}
\end{claim}
\begin{proof}
Note that 
\[
\left \lvert \E_{x \sim \mcl{D}_i}[x] - \E_{x \sim \mcl{D}_\infty}[x] \right \rvert \leq \sum_{j = i}^{\infty} (1- \phi) j \phi^j  \leq \frac{\delta}{1 - \phi} 
\]
and 
\[
\left \lvert \Var_{x \sim \mcl{D}_i}[x] - \Var_{x \sim \mcl{D}_\infty}[x] \right \rvert \leq \sum_{j = i}^{\infty} (1- \phi) j^2 \phi^j \leq \frac{\delta}{(1 - \phi)^2} \,.
\]
\end{proof}

We also have the following useful property that small changes in the scaling parameter $\phi$ only result in a small change in the distribution of the Mallows model. 
\begin{claim}\label{claim:small-TV}
Let $\pi^*$ be any permutation on $n$ elements.  Then for any $0 < \phi, \phi' < 1$ such that $|\phi - \phi'| \leq \eps/n^2$, we have
\[
D_{\TV}( M(\phi, \pi^*) , M(\phi', \pi^*)) \leq \eps \,. 
\]
\end{claim}
\begin{remark}
Quantitatively, this bound is not tight at all but it will suffice for our purposes.
\end{remark}
\begin{proof}
We can imagine sampling $\pi \sim M(\phi, \pi^*)$ and $\pi' \sim M(\phi' , \pi^*)$ using the insertion procedure.  It suffices to upper bound the TV distance between the insertion distributions $\mcl{D}_{i, \phi}$ and $\mcl{D}_{i, \phi'}$ for any $i$.  WLOG $\phi \geq \phi'$.  Then for any $1 \leq i \leq n$, we have
\[
D_{\TV}(\mcl{D}_{i, \phi}, \mcl{D}_{i, \phi'} ) \leq \sum_{j = 1}^{i-1} \frac{\phi^j - \phi'^j}{1 + \phi + \dots + \phi^{i - 1} } \leq \sum_{j = 1}^{i-1} \frac{j\phi^{j-1}| \phi - \phi'|}{1 + \phi + \dots + \phi^{i - 1} } \leq \eps/n \,.
\]
Now, using the insertion procedure, a sample from  $M(\phi, \pi^*)$ is uniquely determined by one draw from each of $\mcl{D}_{1, \phi}, \dots , \mcl{D}_{n, \phi}$ and a sample from  $M(\phi', \pi^*)$ is uniquely determined by one draw from each of $\mcl{D}_{1, \phi'}, \dots , \mcl{D}_{n, \phi'}$.  Thus, combining the previous inequality over all $i$ gives the desired result.
\end{proof}

\subsection{Restrictions and Blocks}

The distributions generated by Mallows models also have nice structural properties when restricting to the induced permutation on certain elements in a consecutive block.  In particular, we have the following structural result from \cite{liu2018efficiently}.
\begin{lemma}[See \cite{liu2018efficiently}] \label{lem:restricted-block}
Let $M(\phi, \pi^*)$ be a Mallows model on $n$ elements.  Let $1 \leq a_1 < b_1 < a_2 < b_2 < \dots < a_k < b_k \leq n$ be integers.   Define $B_i = \{ \pi^*(a_i), \pi^*(a_i+1), \dots , \pi^*(b_i) \}$. Then for a sample $\pi \sim M( \phi, \pi^*)$, the distributions of $\pi|_{B_1}, \dots , \pi|_{B_k}$ are independent and the distribution of $\pi|_{B_i}$ is $M(\phi, \pi^*|_{B_i})$ for all $i$. 
\end{lemma}

\paragraph{Generalized Insertion Procedure:} Lemma \ref{lem:restricted-block} gives us additional ways to sample from a Mallows model that generalize the insertion procedure.  Note that the insertion procedure clearly can be run from the front (starting with the first element of $\pi^*$) or from the back (starting with the last element of $\pi^*$).  Using Lemma \ref{lem:restricted-block} for one block, we can actually sample from a Mallows model $M(\phi, \pi^*)$ using the insertion procedure but starting from the middle.  In particular, we can start with some element $\pi^*(i)$ and at each step we may either insert the next element to the left or next element to the right according to the corresponding insertion distribution (note we always maintain that the elements that we have inserted so far form a contiguous block in $\pi^*$). 

\subsection{Inversion Probability}
The final basic property that we need controls the probability that for a sample $\pi \sim M(\phi, \pi^*)$, two given elements of $\pi^*$ are inverted in the sample $\pi$.  In particular, we show that (up to some threshold) the inversion probability decreases roughly linearly with the distance between the elements in $\pi^*$. 

\begin{claim}\label{claim:inversion-prob}
Let $M(\phi, \pi^*)$ be a Mallows model on $n$ elements.  Consider the elements $\pi^*(a)$ and $\pi^*(b)$ for some $1 \leq a < b \leq n$.  Let $P$ be the probability that $\pi^*(b)$ appears before $\pi^*(a)$ in a random sample from $M(\phi, \pi^*)$.  Then
\[
\frac{1}{2} -  8 (b - a)(1-\phi) \leq P \leq \frac{1}{2} - 0.01 \min(1, (b-a)(1-\phi))
\]
\end{claim}
\begin{proof}
By Lemma \ref{lem:restricted-block}, it suffices to consider when $a = 1$ and $b = d$ and we are sampling from a Mallows model on $d$ elements.  Now note that we can imagine sampling from $M(\phi, \pi^*)$ by first sampling the ordering of the elements $(\pi^*(2), \dots , \pi^*(d-1))$ and then inserting $\pi^*(1)$ into the $d-1$ possible positions with probabilities
\[
\frac{1}{1 + \phi + \dots + \phi^{d-2}}, \frac{\phi}{1 + \phi + \dots + \phi^{d-2}}, \dots , \frac{\phi^{d-2}}{1 + \phi + \dots + \phi^{d-2}}
\]
and then finally inserting $\pi^*(d)$ into the $d$ possible positions with probabilities 
\[
\frac{\phi^{d-1}}{1 + \phi + \dots + \phi^{d-1}}, \frac{\phi^{d-2}}{1 + \phi + \dots + \phi^{d-1}}, \dots , \frac{1}{1 + \phi + \dots + \phi^{d-1}} \,.
\]
By direct computation, the probability that $\pi^*(d)$ occurs before $\pi^*(1)$ in the resulting permutation is 
\begin{align*}
P = \frac{(d-1) \phi^{d-1} + (d-2)\phi^d + \dots + \phi^{2d -3}}{(1 + \phi + \dots + \phi^{d-2}) (1 + \phi + \dots + \phi^{d-1} )}  \,.
\end{align*}
Now, we define
\[
Q = \frac{1}{2} - P = \frac{(d-1)(\phi^{d-2} - \phi^{d-1}) + (d-2)(\phi^{d - 3} - \phi^d ) + \dots + ( 1 - \phi^{2d -3}) }{2(1 + \phi + \dots + \phi^{d-2}) (1 + \phi + \dots + \phi^{d-1} )} \,.
\]
To prove the desired inequality, it suffices to obtain upper and lower bounds on $Q$.  First we prove the LHS of the desired inequality.  Note that we may assume $(b - a)(1 - \phi) = (d-1)(1-\phi) < 1/8$ since otherwise the desired inequality is trivially true.  We now have
\[
Q \leq \frac{( (d-1) \cdot 1 + (d-2) \cdot 3 + \dots + 1 \cdot (2d - 3) ) ( 1 - \phi) }{2d(d-1) \phi^{2d - 3}} \leq \frac{(d - 1)(1-\phi)}{\phi^{2d - 3}} \leq 8(d - 1)(1- \phi) \,.
\]
This immediately gives the LHS of the desired inequality.  Now we prove the RHS.  If $d-1 \leq 1/(1-\phi)$  then we have
\[
Q \geq \frac{0.2 ((d-1) \cdot 1 + (d-2) \cdot 3 + \dots + 1 \cdot (2d-3)) ( 1 - \phi) \cdot }{2d(d-1)} \geq 0.01 ( d-1)(1 - \phi) \,.
\]
Otherwise, assume $d - 1 \geq 1/(1 - \phi)$.  Let $c = \lceil 1/(2(1 - \phi)) \rceil $.  We have 
\[
Q \geq  \frac{(1 - \phi)^2 ( (1 - \phi^{2d-3}) + 2(\phi - \phi^{2d - 4}) + \dots + c ( \phi^{c-1} - \phi^{2d - 2 - c}) ) }{2} \geq 0.01 \,.
\]
Putting the above together, we get the RHS of the desired inequality and we are done.
\end{proof}
\section{Position Vector to TV Distance}\label{sec:position-to-TV}

The first important insight is that the TV distance between two Mallows models is equivalent, up to constant factors, to the $L_2$ distance between their position vectors.  The main result of this section is stated below.

\begin{theorem}\label{thm:TV-characterization}
For any two Mallows models on $n$ elements $M(\phi, \pi^*)$ and $M(\phi, \sigma^*)$, we have
\[
D_{\TV}\left( M(\phi, \pi^*), M(\phi, \sigma^*)\right) \leq 2 (1 - \phi) \norm{v_{\pi^*} - v_{\sigma^*}}\,.
\]
\end{theorem}

The above theorem says that $L_2$ distance between position vectors upper bounds the TV distance.  The next theorem is the corresponding lower bound.  Only the upper bound will be used in our learning algorithm but we believe the tight characterization of TV distance in terms of $L_2$ distance is of independent interest.  In fact, the analysis of our learning algorithm already implies that the relation between TV and $L_2$ distance is tight up to logarithmic factors.  In Theorem \ref{thm:TV-lowerbound}, through a direct approach, we prove that this relationship is tight up to constant factors.

\begin{theorem}\label{thm:TV-lowerbound}
There is a (sufficiently large) universal constant $C$ such that the following holds.  For any two Mallows models on $n$ elements $M(\phi, \pi^*)$ and $M(\phi, \sigma^*)$ such that 
\[
\norm{v_{\pi^*} - v_{\sigma^*}} \leq \frac{1}{2(1 - \phi)}
\]
we have 
\[
D_{\TV}\left( M(\phi, \pi^*), M(\phi, \sigma^*)\right) \geq (1 - \phi) \norm{v_{\pi^*} - v_{\sigma^*}}/C \,.
\]
\end{theorem}
\begin{remark}
Note that the above characterization can only hold when  $\norm{v_{\pi^*} - v_{\sigma^*}} = O(1/(1-\phi))$ since otherwise the TV distance lower bound would be larger than $1$.  Thus, some upper bound on $\norm{v_{\pi^*} - v_{\sigma^*}}$ is necessary.  In our proof we do not worry about optimizing the constants in these bounds.  
\end{remark}
The remainder of this section is devoted to proving Theorem \ref{thm:TV-characterization}.  The proof of Theorem \ref{thm:TV-lowerbound} is deferred to Appendix \ref{appendix:TV-lowerbound}.  To prove Theorem \ref{thm:TV-characterization}, we will rely on bounding the KL divergence between the distributions $M(\phi, \pi^*)$ and $M(\phi, \sigma^*)$ and then using Pinsker's inequality.  We first introduce some notation.
\begin{definition}\label{def:distance-diff}
For permutations $\pi, \pi^*, \sigma^*$ on $n$ elements, define
\[
\Delta_{\pi^*, \sigma^*}(\pi) = d_{\KT}(\pi, \pi^*) - d_{\KT}(\pi, \sigma^*) \,.
\]
\end{definition}
\begin{definition}\label{def:difference-set}
For two permutations $\pi, \sigma$ on $[n]$, let $S(\pi, \sigma)$ be the set of ordered pairs $(x,y)$ such that $x$ occurs before $y$ in $\pi$ and $y$ occurs before $x$ in $\sigma$.
\end{definition}
\begin{remark}
Note that $S(\sigma, \pi)$ is exactly $S(\pi, \sigma)$ with the order of all pairs reversed.
\end{remark}

For a sample $ \pi \sim M(\phi, \pi^*)$ we have
\begin{equation}\label{eq:ratio}
\frac{\Pr_{M(\phi, \pi^*)}[\pi]}{\Pr_{M(\phi, \sigma^*)}[\pi]} = \phi^{d_{\KT}(\pi, \pi^*) - d_{\KT}(\pi, \sigma^*)} = \phi^{\Delta_{\pi^*,\sigma^*}(\pi)}\,.
\end{equation}
Thus, to compute the KL divergence, it suffices to analyze the expectation of the quantity $\Delta_{\pi^*,\sigma^*}(\pi)$ for $\pi \sim M(\phi, \pi^*)$.  First, we prove a combinatorial identity.

\begin{claim}\label{claim:sum-identity}
For two permutations $\pi$ and $\sigma$ on $[n]$, we have
\[
\sum_{(x,y) \in S(\pi, \sigma)} \left(\pi^{-1}(y) - \pi^{-1}(x) \right) = \frac{1}{2} \norm{v_\pi - v_{\sigma}}^2 \,.
\]
\end{claim}
\begin{proof}
WLOG $\pi$ is the identity permutation.  We proceed by induction on $n$.  The base case is trivial.  If $\sigma(n) = n$ then clearly we can delete element $n$ from both permutations and use the induction hypothesis to finish.  Now consider when $\sigma(n) \neq n$.  Let $\sigma^{-1}(n) = n - d$ for some $1 \leq d \leq n-1 $.
\\\\
We now evaluate what happens to each of the two sides when element $n$ is deleted from both permutations $\pi$ and $\sigma$.  We have
\begin{align*}
\sum_{(x,y) \in S(\pi, \sigma)} \left(\pi^{-1}(y) - \pi^{-1}(x) \right) &= \sum_{(x,y) \in S(\pi|_{[n-1]}, \sigma|_{[n-1]})} \left( \left(\pi|_{[n-1]} \right) ^{-1}(y) - \left(\pi|_{[n-1]} \right)^{-1}(x) \right)  + \sum_{i = 0}^{d-1} \left( n - \sigma(n - i ) \right) \\ &= \sum_{(x,y) \in S(\pi|_{[n-1]}, \sigma|_{[n-1]})} \left( \left(\pi|_{[n-1]} \right) ^{-1}(y) - \left(\pi|_{[n-1]} \right)^{-1}(x) \right) + dn - \sum_{i = 0}^{d-1} \sigma(n - i)\,.
\end{align*}
We also have
\begin{align*}
\norm{v_\pi - v_{\sigma}}^2 &= d^2 + \norm{\pi|_{[n-1]} - v_{\sigma|_{[n-1]}}}^2 + \sum_{i = 0}^{d-1} \left( (\sigma(n - i) - (n - i))^2 - (\sigma(n - i) - 1 - (n - i) )^2 \right) \\ &= \norm{\pi|_{[n-1]} - v_{\sigma|_{[n-1]}}}^2 + d^2 + \sum_{i = 0}^{d-1} \left(2(\sigma(n - i) - (n- i)) - 1 \right) \\ &= \norm{\pi|_{[n-1]} - v_{\sigma|_{[n-1]}}}^2  + 2dn - 2\sum_{i = 0}^{d-1} \sigma(n - i)
\end{align*}
Using the induction hypothesis on the permutations $\pi|_{[n-1]}$ and $ \sigma|_{[n-1]}$, we conclude
\[
\sum_{(x,y) \in S(\pi, \sigma)} \left(\pi^{-1}(y) - \pi^{-1}(x) \right) = \frac{1}{2} \norm{v_\pi - v_{\sigma}}^2
\]
as desired.
\end{proof}

Next, we relate the quantity $\norm{v_{\pi^*} - v_{\sigma^*}}$ to the KL divergence between the distributions $M(\phi, \pi^*)$ and $M(\phi, \sigma^*)$ after which Theorem \ref{thm:TV-characterization} will follow immediately from Pinsker's inequality.

\begin{lemma}\label{lem:KL-characterization}
For any two Mallows models on $n$ elements $M(\phi, \pi^*)$ and $M(\phi, \sigma^*)$, we have
\[
D_{\KL}\left( M(\phi, \pi^*) \| M(\phi, \sigma^*)\right) \leq 4 (1 - \phi) (- \log \phi) \norm{v_{\pi^*} - v_{\sigma^*}}^2\,.
\]
\end{lemma}
\begin{proof}
We may write
\begin{equation}\label{eq:KL}
D_{\KL}\left( M(\phi, \pi^*) \| M(\phi, \sigma^*)\right) = \E_{\pi \sim M(\phi, \pi^*)} \left[ \log \frac{\Pr_{M(\phi, \pi^*)}[\pi]}{\Pr_{M(\phi, \sigma^*)}[\pi]}\right] = (\log \phi) \cdot  \E_{\pi \sim M(\phi, \pi^*)}\left[\Delta_{\pi^*,\sigma^*}(\pi) \right] \,.
\end{equation}
To evaluate the expectation of $\Delta_{\pi^*,\sigma^*}(\pi)$, we may use linearity of expectation and sum over all pairs $(x,y)$ with $1 \leq x < y \leq n$.  Specifically, note that 
\[
d_{\KT}(\pi, \pi^*) = \frac{1}{2}\sum_{1 \leq x < y \leq n} 1 - \left(\sign(\pi^{-1}(x) - \pi^{-1}(y) ) \sign(\pi^{*-1}(x) - \pi^{*-1}(y) ) \right)
\]
and similar for $d_{\KT}(\pi, \sigma^*)$. Thus we have
\begin{align*}
\E\left[ \Delta_{\pi^*,\sigma^*}(\pi)  \right] = \frac{1}{2}\sum_{1 \leq x < y \leq n}   \E\left[ \sign(\pi^{-1}(x) - \pi^{-1}(y) ) \left(\sign(\sigma^{*-1}(x) - \sigma^{*-1}(y) )- \sign(\pi^{*-1}(x) - \pi^{*-1}(y))   \right) \right]  \,.
\end{align*}
For all pairs with $(x,y) \notin S(\pi^*, \sigma^*)$, i.e. those pairs for which $x$ and $y$ occur in the same order in $\pi^*$ and $\sigma^*$, it is clear that the difference in the summand above is always $0$.  Thus, it actually suffices to sum over $(x,y) \in S(\pi^*, \sigma^*)$.  Next, consider $(x,y) \in S(\pi^*, \sigma^*)$.  By definition, we must have 
\[
\sign(\sigma^{*-1}(x) - \sigma^{*-1}(y) ) = - \sign(\pi^{*-1}(x) - \pi^{*-1}(y)) \,.
\]
Thus, by Claim \ref{claim:inversion-prob}, for $(x,y) \in S(\pi^*,\sigma^*)$, we have
\begin{align*}
&  \E\left[ \sign(\pi^{-1}(x) - \pi^{-1}(y) ) \left(\sign(\sigma^{*-1}(x) - \sigma^{*-1}(y) )- \sign(\pi^{*-1}(x) - \pi^{*-1}(y))   \right) \right]  \\ &= 1 - 2 \cdot \Pr\left[ \sign(\pi^{-1}(x) - \pi^{-1}(y) ) = \sign(\pi^{*-1}(x) - \pi^{*-1}(y) ) \right]  \\ & \geq -16 \cdot  \left( \pi^{*-1}(y) - \pi^{*-1}(x) \right) \cdot  ( 1- \phi) \,. 
\end{align*}
We now conclude
\[
\E_{\pi \sim M(\phi, \pi^*)}\left[ \Delta_{\pi^*,\sigma^*}(\pi) \right]  \geq -\sum_{(x,y) \in S(\pi^*, \sigma^*)} 8 \cdot  \left( \pi^{*-1}(y) - \pi^{*-1}(x) \right) \cdot  ( 1- \phi) = -4(1-\phi) \norm{v_{\pi^*} - v_{\sigma^*}}^2 
\]
where the last step follows from Claim \ref{claim:sum-identity}.  Plugging the above into (\ref{eq:KL}) (recall that $\log \phi < 0$ so the direction of the inequality flips) immediately gives the desired result.
\end{proof}

\begin{proof}[Proof of Theorem \ref{thm:TV-characterization}]
First if $\phi \leq 1/2$ then the inequality is trivially true (since if $v_{\pi^*} = v_{\sigma^*}$ then the TV distance is clearly $0$ and otherwise the RHS is greater than $1$).  

Now assume $\phi \geq 1/2$.  By Pinsker's inequality and Lemma \ref{lem:KL-characterization}, we get
\begin{align*}
D_{\TV}\left( M(\phi, \pi^*), M(\phi, \sigma^*)\right) \leq  \sqrt{\frac{1}{2}D_{\KL}\left( M(\phi, \pi^*) \| M(\phi, \sigma^*)\right)} \leq \sqrt{2(1 - \phi) ( - \log \phi) } \norm{v_{\pi^*} - v_{\sigma^*}} \\ \leq 2(1 - \phi)\norm{v_{\pi^*} - v_{\sigma^*}}
\end{align*}
as desired.
\end{proof}

\section{Tail Bounds on the Position Vector Distribution}\label{sec:tail-bounds}

Before we present our learning algorithm, we will need to understand some additional properties of the distribution generated by a Mallows model.  Our goal will be to robustly estimate the mean of $v_{\pi}$ for  $\pi \sim M(\phi, \pi^*) $.  In order to do this, we will first need to prove some properties about the tail decay of the distribution of $v_{\pi}$.  This is done in Lemma \ref{lem:Gaussian-tails}, the main result of this section.  Before we can state the lemma, we will need some additional definitions.

It will actually be useful to split $v_{\pi}$ into two vectors that we define below.  Recall Definitions \ref{def:front} and \ref{def:back}.
\begin{definition}\label{def:forward-component2}
For permutations $\pi, \wh{\pi}$ on $[n]$ define the $n$-dimensional vector $v_{\pi}^{\leftarrow \wh{\pi}}$ as follows.  The $i$\ts{th} entry is equal to 
\[
\wh{\pi}^{-1}(i) - \left(\pi_{\wh{\pi}: \wh{\pi}^{-1}(i) } \right)^{-1}(i) \,.
\]
In other words, it is equal to the number of elements among $\{\wh{\pi}(1),\wh{\pi}(2), \dots , \wh{\pi}(\wh{\pi}^{-1}(i) - 1) \}$ that occur after $i$ in $\pi$.  We call $v_{\pi}^{\leftarrow \wh{\pi}}$ the front adjustment of $v_{\pi}$ with respect to $\wh{\pi}$.
\end{definition}
\begin{definition}\label{def:backward-component2}
For permutations $\pi, \wh{\pi}$ on $[n]$ define the $n$-dimensional vector $v_{\pi}^{\rightarrow \wh{\pi}}$ as follows.  The $i$\ts{th} entry is equal to 
\[
 \left(\pi_{\wh{\pi}:-(n - \wh{\pi}^{-1}(i) + 1) } \right)^{-1}(i) - 1 \,.
\]
In other words, it is equal to the number of elements among $\{\wh{\pi}(\wh{\pi}^{-1}(i)+1), \dots , \wh{\pi}(n) \}$ that occur before $i$ in $\pi$.  We call $v_{\pi}^{\rightarrow \wh{\pi}}$ the back adjustment of $v_{\pi}$ with respect to $\wh{\pi}$.
\end{definition}

We will abbreviate the above notation when $\wh{\pi} = \id$.
\begin{definition}\label{def:abbreviate}
We use the shorthand $v_{\pi}^{\leftarrow} = v_{\pi}^{\leftarrow \id} $ and $v_{\pi}^{\rightarrow} = v_{\pi}^{\rightarrow \id}$.
\end{definition}

First, we have the following simple observation.
\begin{claim}\label{claim:split-vector}
For any permutations $\pi, \wh{\pi}$ on $[n]$, we have
\begin{equation}\label{eq:addition}
v_{\pi}  = v_{\wh{\pi}} +     v_{\pi}^{\rightarrow \wh{\pi}} - v_{\pi}^{\leftarrow \wh{\pi}} \,.
\end{equation}
\end{claim}
\begin{proof}
The desired relation follows immediately from the definitions of $ v_{\pi}^{\leftarrow \wh{\pi}}, v_{\pi}^{\rightarrow \wh{\pi}} $.
\end{proof}

Using the insertion model for sampling (Lemma \ref{lem:insertion-sampling}), we see that the entries of $v_{\pi}^{\leftarrow \pi^*}$ are actually independent for $\pi \sim M(\phi, \pi^*)$ (and similar for $v_{\pi}^{\rightarrow \pi^*}$).  Thus, we can obtain tail bounds for each of $v_{\pi}^{\leftarrow \pi^*}$ and $v_{\pi}^{\rightarrow \pi^*}$ separately.  The key bound, that the tail decay of the distributions in any direction is sub-exponential, is stated below.

\begin{lemma}\label{lem:Gaussian-tails}
Let $M(\phi, \pi^*)$ be a Mallows model on $n$ elements and assume $\phi \geq 0.5$. Let $v$ be any unit vector.  Then for any real number $t \geq 0$
\[
\Pr_{\pi \sim M(\phi, \pi^*)}\left[ \left \lvert v \cdot v_{\pi}^{\leftarrow \pi^*} - \E[v \cdot v_{\pi}^{\leftarrow \pi^*}] \right \rvert \geq  \frac{t}{1-  \phi}  \right] \leq 4e^{-0.2t}  
\]
and the same concentration inequality holds for $v_{\pi}^{\rightarrow \pi^*}$.
\end{lemma}

The proof of Lemma \ref{lem:Gaussian-tails} is deferred to Appendix \ref{appendix:tail-bounds}.  Note that combining Claim \ref{claim:split-vector} and Lemma \ref{lem:Gaussian-tails} immediately gives sub-exponential tail bounds for the distribution of $v_{\pi}$.



\section{Learning Algorithm}\label{sec:main-alg}

\noindent We are now ready to describe our full learning algorithm.  

\subsection{Basic Reductions}\label{sec:reductions}
We will first introduce some additional notation that will be used throughout this section and make a few basic reductions.  
\paragraph{Notation:}  
\begin{itemize}
    \item We receive $\eps$-corrupted samples from a Mallows model $M(\phi, \pi^*)$  on $n$ elements $\{1,2, \dots , n \}$
    \item Set $\theta = \lceil 10^3\log(n/\eps)/(1 - \phi) \rceil $
\end{itemize}

\paragraph{Assumptions:} 
\begin{enumerate}
    \item Assume $\eps$ is smaller than some constant, say $\eps < 0.1$ and $\phi > 1/2$ 
    \item $\phi$ is known and $\phi \leq 1 - \eps/n^2$
    \item $\pi^*$ is actually a permutation on $n + 2\theta$ elements whose first $\theta$ elements are $1,2, \dots , \theta$ and the last $\theta$ elements are $n + \theta + 1, \dots , n + 2\theta$ (and the middle $n$ elements can be permuted arbitrarily)
    \item We actually receive samples from this extended Mallows model $M(\phi, \pi^*)$ on $n + 2\theta$ elements
\end{enumerate}

We now justify why we can make these assumptions.  The first assumption is clearly valid because if $\phi < 1/2$, then we can trivially learn the relative order of every pair of elements (by Claim \ref{claim:inversion-prob}).  It remains to justify the other two.

\subsubsection{Guessing the Scaling Parameter}
We first guess $\phi$.  In particular, we try guessing $\wh{\phi} = 0, \eps/n^2, 2\eps/n^2, \dots , 1 $.  One of these guesses will be within $\eps/n^2$ of the true $\phi$ and thus by Claim \ref{claim:small-TV}, we can pretend that we are actually receiving a $3\eps$-corrupted sample from the distribution $M(\wh{\phi}, \pi^*)$.  Thus, we can pretend that we know the true scaling parameter $\phi$ (only losing a constant factor in the corruption rate).  While we do enumerate over multiple guesses for $\phi$, at the end, we can simply run a hypothesis test (see Section \ref{sec:hypothesis-test}) to rule out all extraneous hypotheses and ensure that our final output is actually $\wt{O}(\eps)$-close in TV distance to the true distribution.  Furthermore, if $\phi = 1$, then we do not need to learn the permutation since the distribution is just uniform so we may assume that $\phi \leq 1 - \eps/n^2$.

\subsubsection{Padding}\label{sec:padding}
We claim that we can extend $\pi^*$ to a permutation on $n + 2\theta$ elements whose first $\theta$ elements are $1, 2, \dots , \theta$ and the last $\theta$ elements are $n + \theta +1, \dots , n + 2\theta$.  To see why this is true, we can imagine padding $\pi^*$ with $\theta$ dummy elements at the beginning and at the end. We can then insert these dummy elements into each of our samples according to the insertion process.  This allows us to take $n \leftarrow n + 2\theta$ and simulate $\eps$-corrupted samples from a Mallows model on $n + 2\theta$ elements (the $n$ original elements plus the $2 \theta$ dummy elements) where we know the first and last $\theta$ elements of the base permutation. 

Recall that we have defined vectors $v_{\pi}, v_{\pi}^{\leftarrow}, v_{\pi}^{\leftarrow \wh{\pi}}$ in previous sections.  However, since we now assume that $\pi^*$ matches the identity permutation on the first and last $\theta$ elements, we will only focus on the middle $n$ entries of these vectors.  In light of this, we use the following notation throughout the remainder of this section.
\begin{definition}
For a vector $v \in \R^{n + 2\theta}$, we use $\trunc(v)$ to denote the vector in $\R^{n}$ formed by deleting the first and last $\theta$ entries of $v$. 
\end{definition}

To see why the padding is useful, imagine $\pi^* = \id$.  note that the distribution of the $i$\ts{th} entry of $v_{\pi}^{\leftarrow}$ for $\pi \sim M(\phi, \id)$ is $\mcl{D}_i$ (recall Definition \ref{def:insertion-distribution}) and furthermore the entries are independent.  For small $i$, the distribution $\mcl{D}_i$ is not close to $\mcl{D}_{\infty}$ and in particular its variance is much smaller.  However, once we truncate, all of the remaining entries will have distribution close to $\mcl{D}_{\infty}$ and the covariance of the distribution of $\trunc\left( v_{\pi}^{\leftarrow}\right)$ will be close to $\phi/(1 - \phi)^2I $.  This property will be convenient for applying existing tools for robust mean estimation.

\subsection{Algorithm Description}

At a high-level, our algorithm consists of two parts: 
\begin{itemize}
    \item Obtaining a rough estimate for the unknown permutation (see {\sc Rough Permutation Estimate})
    \item Iteratively refining this estimated permutation (see {\sc Iterative Refinement})
\end{itemize} 

\noindent Our full algorithm is summarized in the blocks below.  In it, we use black-box results for robustly estimating the mean of distributions with bounded covariance and near-identity covariance.  These routines are now standard in robust statistics literature and their guarantees are stated formally in the next section (see Corollary \ref{coro:est-with-bounded-cov} and Theorem \ref{thm:est-with-estimated-cov}).  However, the details are not important for understanding the overall algorithm, which is explained below.
\begin{algorithm}[H]
\caption{{\sc Full Algorithm} }
\begin{algorithmic} 
\State \textbf{Input:} Parameter $\phi$
\State \textbf{Input:} $\eps$-corrupted sample $\pi_1, \dots , \pi_s $ from unknown Mallows model $M(\phi, \pi^*)$
\State Run {\sc Rough Permutation Estimate} to obtain estimate $\wh{\pi}$ for $\pi$
\For{$t = 1,2, \dots , \lceil 10  \log 1/\eps \rceil$}
\State Run {\sc Iterative Refinement} using estimate $\wh{\pi}$ to obtain new estimate $\wt{\pi}$
\State Set $\wh{\pi} \leftarrow \wt{\pi}$
\EndFor
\State \textbf{Output:} $\wt{\pi}$
\end{algorithmic}
\label{alg:full}
\end{algorithm}

\begin{algorithm}[H]
\caption{{\sc Rough Permutation Estimate} }
\begin{algorithmic} 
\State \textbf{Input:} Parameter $\phi$
\State \textbf{Input:} $\eps$-corrupted sample $\pi_1, \dots , \pi_s $ from unknown Mallows model $M(\phi, \pi^*)$
\State  Use Corollary \ref{coro:est-with-bounded-cov} (robust mean estimation with bounded covariance) to get estimate $\wh{v}$ for  
\[
\frac{(1 - \phi)}{\sqrt{\phi}} \E_{\pi \sim M(\phi, \pi^*)}[\trunc(v_{\pi})] \,.
\]
\State Let $\wh{\pi}$ be the permutation on $[n + 2\theta]$ that matches $\id$ on the first and last $\theta$ elements and sorts the entries of $\wh{v}$, namely 
\[
\wh{v}[\wh{\pi}(\theta + 1) - \theta] \leq \dots \leq \wh{v}[\wh{\pi}(n + \theta) - \theta] \,.
\]
\State \textbf{Output:} $\wh{\pi}$
\end{algorithmic}
\label{alg:rough-estimate}
\end{algorithm}

\begin{algorithm}[H]
\caption{{\sc Iterative Refinement} }
\begin{algorithmic} 
\State \textbf{Input:} Parameter $\phi$
\State \textbf{Input:} $\eps$-corrupted sample $\pi_1, \dots , \pi_s $ from unknown Mallows model $M(\phi, \pi^*)$
\State \textbf{Input:} Permutation estimate $\wh{\pi}$ 
\State Use Theorem \ref{thm:est-with-estimated-cov} (robust mean estimation with near-identity covariance) to obtain estimates $\wh{v^{\leftarrow}},\wh{v^{\rightarrow}} $ for 
\[
\frac{(1 - \phi)}{\sqrt{\phi}} \E_{\pi \sim M(\phi, \pi^*)}\left[\trunc\left(v_{\pi}^{\leftarrow \wh{\pi}}\right)\right] ,  \frac{(1 - \phi)}{\sqrt{\phi}} \E_{\pi \sim M(\phi, \pi^*)}\left[\trunc\left(v_{\pi}^{\rightarrow \wh{\pi}}\right)\right]  \,.
\]
\State Let $\wh{v} = (1 - \phi)/\sqrt{\phi} \cdot  \trunc(v_{\wh{\pi}})  - \wh{v^{\leftarrow}} + \wh{v^{\rightarrow}}$
\State  Let $\wt{\pi}$ be the permutation on $[n + 2\theta]$ that matches $\id$ on the first and last $\theta$ elements and sorts the entries of $\wh{v}$, namely
\[
\wh{v}[\wt{\pi}(\theta + 1) - \theta] \leq \dots \leq \wh{v}[\wt{\pi}(n + \theta) - \theta] \,.
\]
\State \textbf{Output:} $\wt{\pi}$
\end{algorithmic}
\label{alg:iterative-refinement}
\end{algorithm}
\begin{remark}
Note that the $(1 - \phi)/\sqrt{\phi}$ normalization factor is because the entries of the vectors we are estimating have variance $\sim \phi/(1 - \phi)^2$ so this scaling will make the covariance matrix comparable to the identity matrix.
\end{remark}

The main theorem about the guarantees of our algorithm is stated below.
\begin{theorem}\label{thm:analysis-main}
Assume that $s \geq (n/\eps)^2 \poly(\log (n/\eps))$ and that Assumptions $1,2,3$ hold.  Then with probability at least $ 1 - (\eps/n)^{9}$, the output $\wt{\pi}$ of {\sc Full Algorithm} satisfies
\[
\norm{v_{\wt{\pi}} - v_{\pi^*} }  \leq O\left(  \frac{\eps \log 1/\eps }{1 - \phi}\right)\,.
\]
\end{theorem}

The proof of Theorem \ref{thm:analysis-main} relies on the following two lemmas about the guarantees of the subroutines.  The first is for our rough permutation estimate and the second guarantees that we make significant progress in each refinement step until we are sufficiently close to the true permutation.

\begin{lemma}\label{lem:rough-estimate}
Assume that $s \geq (n/\eps)^2 \poly(\log (n/\eps))$ and that Assumptions $1,2,3$ hold.  Then with probability at least $ 1 -(\eps/n)^{10}$, the output $\wh{\pi}$ of {\sc Rough Permutation Estimate} satisfies
\[
\norm{v_{\wh{\pi}} - v_{\pi^*} } \leq  \frac{C \sqrt{\eps}}{1 - \phi} 
\]
for some sufficiently large universal constant $C$.
\end{lemma}

\begin{lemma}\label{lem:refinement}
Assume that $s \geq (n/\eps)^2 \poly(\log (n/\eps))$ and that Assumptions $1,2,3$ hold.   Then with probability at least $ 1 - (\eps/n)^{10}$ over the input sample, the following holds: for any input estimate $\wh{\pi}$  to {\sc Iterative Refinement} that matches $\id$ on the first and last $\theta$ elements and satisfying $ \norm{v_{\wh{\pi}} - v_{\pi^*} }  \leq 1/(10 ( 1 - \phi))$,  the output $\wt{\pi}$ of {\sc Iterative Refinement} satisfies
\begin{equation}\label{eq:recursion}
\norm{v_{\wt{\pi}} - v_{\pi^*} } \leq C\left(  \sqrt{ \frac{\eps  \norm{v_{\wh{\pi}} - v_{\pi^*} }}{1 - \phi} } + \frac{\eps \log 1/\eps}{1 - \phi} \right) 
\end{equation}
for some sufficiently large universal constant $C$
\end{lemma}

Theorem \ref{thm:analysis-main} will follow by simply combining the guarantees of Lemma \ref{lem:rough-estimate} and Lemma \ref{lem:refinement} and analyzing the recurrence given by (\ref{eq:recursion}).

\subsection{Stability and Robust Mean Estimation Preliminaries}\label{sec:robust-mean-est}

We now formally introduce the machinery for robust mean estimation that we will use.  First, we define stability, which is a well-studied concept that is often useful in the context of robust estimation algorithms.  The survey \cite{diakonikolas2019recent} explains how robust learning algorithms are often derived from stability bounds.   

\begin{definition}
For $\eps > 0$ and $\delta \geq \eps$, we say a finite set $S \subset \R^d$ is $(\eps, \delta)$-stable with respect to a distribution $\mcl{D}$ with mean $\mu_{\mcl{D}}$ if for every unit vector $v \subset \R^d$ and every subset $S' \subseteq S$ of size at least $(1- \eps)|S|$ we have
\begin{align*}
&\left \lvert v \cdot \left(\mu_{\mcl{D}} - \frac{1}{S'}\sum_{x \in S'} x \right)  \right \rvert    \leq \delta \\
&\left \lvert \frac{1}{S'}\sum_{x \in S'}(v \cdot (x - \mu_{\mcl{D}}))^2 - 1  \right \rvert    \leq  \frac{\delta^2}{\eps}
\end{align*}
\end{definition}

Roughly, a set of samples $S$ is stable if any $1-\eps$-fraction of the samples has empirical mean close to the mean of $\mcl{D}$ and empirical covariance close to the identity.  We will rely on the following results from \cite{diakonikolas2019recent} about robustly estimating the mean of a distribution via stability.

\begin{theorem}[Theorem 2.7 in \cite{diakonikolas2019recent}]\label{thm:est-with-estimated-cov}
Let $S$ be a $(3\eps , \delta)$-stable set with respect to a distribution $X$ and let $T$ be an $\eps$-corrupted version of $S$.  There is a polynomial time algorithm which given $T$ returns $\wh{\mu}$ such that 
\[
\norm{\wh{\mu} - \mu_X} = O(\delta) \,.
\]
\end{theorem}
\begin{corollary}[Corollary 2.9 in \cite{diakonikolas2019recent}]\label{coro:est-with-bounded-cov}
Let $T$ be an $\eps$-corrupted set of samples of size at least $(d/\eps) \poly(\log (d/\eps))$ from a distribution $X$ on $\R^d$ with unknown bounded covariance $\Sigma_X \preceq \sigma^2 I$.  There exists a polynomial time algorithm which given $T$ returns $\wh{\mu}$ such that with $1 - (\eps/d)^{10}$ probability 
\[
\norm{\wh{\mu} - \mu_X} = O(\sigma\sqrt{\eps}) \,.
\]
\end{corollary}

\subsubsection{Generic Stability Bounds}

We will rely on the following generic stability bound to eventually prove stability for the distributions of $v_{\pi}, v_{\pi}^{\leftarrow \wh{\pi}}, v_{\pi}^{\rightarrow \wh{\pi}}$ that are of interest in our algorithm.
\begin{corollary}\label{coro:high-dim-stability}
Let $\mcl{D}$ be a distribution on $\R^n$ with mean $\mu$ and covariance $\Sigma$ and assume that it has the property that for all $t > 1$ and any unit vector $v \in \R^n$,
\[
\Pr_{x \sim \mcl{D}}[ | v \cdot (x - \mu)| \geq t ] \leq e^{-t} \,. 
\]
Consider drawing a set of samples $S$ i.i.d from $\mcl{D}$ with  $|S| \geq (n/\eps)^{2}\poly(\log (n/\eps))$.  Then with probability $1 - 2^{-10n \log n /\eps}$, any subset $S' \subset S$ with $|S'| \geq (1-\eps)|S|$ satisfies the following properties for all unit vectors $v \in \R^n$,
\begin{align*}
\left \lvert v \cdot \left( \mu - \frac{1}{|S|'} \sum_{x \in S'} x \right) \right \rvert \leq C \eps \log 1/\eps \\
\left \lvert v^T \Sigma v - \frac{1}{|S|'} \sum_{x \in S'} (v \cdot ( x - \mu))^2 \right \rvert \leq C \eps \log^2 1/\eps \,.
\end{align*}
where $C$ is some sufficiently large universal constant.
\end{corollary}

Since the proof is fairly standard, it is deferred to Appendix \ref{appendix:main-alg}.

\subsection{Analysis of Rough Permutation Estimate}\label{sec:rough-estimate}

We can now proceed with the analysis of our main algorithm.  We first prove Lemma \ref{lem:rough-estimate}.  In light of the guarantees of Corollary \ref{coro:est-with-bounded-cov}, it suffices to upper bound (in spectral norm) the covariance of the distribution of $\trunc(v_{\pi})$ for $\pi \sim M(\phi, \pi^*)$.

\begin{claim}\label{claim:cov-upper-bound}
For any permutation $\pi^*$, we have
\[
\Cov_{\pi \sim M(\phi, \pi^*)}( \trunc(v_{\pi})) \preceq \frac{10^4 }{(1 - \phi)^2} I \,.
\]
\end{claim}
\begin{proof}
WLOG, we may assume $\pi^* = \id$.  Also it suffices to upper bound the covariance of $\Cov_{\pi \sim M(\phi, \id)}( v_{\pi})$ since $\trunc(v_{\pi})$ simply restricts to a subset of the entries of $v_{\pi}$. Now by Claim \ref{claim:split-vector}, we can write
\[
v_{\pi} = v_{\id} - v_{\pi}^{\leftarrow} + v_{\pi}^{\rightarrow} \,.
\]
Note that $v_{\id}$ is constant so
\[
\Cov_{\pi \sim M(\phi, \id)}[v_{\pi} ] \preceq 2 \Cov( v_{\pi}^{\leftarrow}) + 2 \Cov( v_{\pi}^{\rightarrow}) \,. 
\]
Now using Lemma \ref{lem:Gaussian-tails} and integrating, we get that for any unit vector $v$,
\[
v^T \Cov( v_{\pi}^{\leftarrow}) v \leq \frac{10^3}{1 - \phi}
\]
and similar for $\Cov( v_{\pi}^{\rightarrow})$.  Combining all of these bounds, we deduce
\[
\Cov_{\pi \sim M(\phi, \id)}( v_{\pi}) \preceq \frac{10^4 }{(1 - \phi)^2} I
\]
and we are done.
\end{proof}

We will need one more simple observation before we can complete the proof of Lemma \ref{lem:rough-estimate}.  We prove that if a vector $v$ is close to $\E_{\pi \sim M(\phi, \pi^*)}[ v_{\pi}]$ in $L_2$ norm, then the permutation that sorts the entries of $v$ must also be close to $v_{\pi}$.
\begin{claim}\label{claim:sorting}
Let $v \in \R^{n + 2\theta}$ be a vector such that 
\[
\norm{ v - \E_{\pi \sim M(\phi, \pi^*)}[ \trunc(v_{\pi})] } \leq \frac{\delta}{1 - \phi}
\]
for some $\delta > 0$.  Let $\wh{\pi}$ be the permutation that matches $\id$ on the first and last $\theta$ elements and such that 
\[
v[\wh{\pi}(\theta + 1) - \theta] < \dots < v[\wh{\pi}(n + \theta) - \theta] \,.
\]
Then
\[
\norm{v_{\wh{\pi}} - v_{\pi^*} } \leq \frac{2 \delta + 2 \eps}{1 - \phi} \,. 
\]
\end{claim}
\begin{proof}
Recall that for any permutation $\pi$, 
\[
v_{\pi} = v_{\pi^*} + v_{\pi}^{\rightarrow \pi^*} - v_{\pi}^{\leftarrow \pi^*}
\]
and thus the same identity holds after truncating the first and last $\theta$ entries.  Also, note that we have explicit expressions for the distributions of the entries of  $v_{\pi}^{\rightarrow \pi^*},  v_{\pi}^{\leftarrow \pi^*}$ for $\pi \sim M(\phi, \pi^*)$ from insertion sampling.  By Claim \ref{claim:almost-exact-expressions} (and also using the definition of $\theta$), we have
\[
\norm{\E_{\pi \sim M(\phi, \pi^*)} [\trunc(v_{\pi}) ]  - \trunc(v_{\pi^*}) } \leq \frac{\eps}{1 - \phi} \,.
\]
Also, note that since $\wh{\pi}$ is the permutation that sorts the entries of $v$, we have that 
\[
\norm{v - \trunc(v_{\wh{\pi}})} \leq \norm{ v - \trunc(v_{\pi^*})} \leq \frac{\delta + \eps}{1 - \phi} \,.
\]
Combining the above using the triangle inequality we get
\[
\norm{v_{\wh{\pi}} - v_{\pi^*} } = \norm{\trunc(v_{\wh{\pi}}) - \trunc(v_{\pi^*}) }  \leq \frac{2 \delta + 2\eps}{1 - \phi}
\]
where we used the fact that $\pi^*$ and $\wh{\pi}$ both match $\id$ on the first and last $\theta$ elements.  This completes the proof.
\end{proof}

Combining Claim \ref{claim:cov-upper-bound} and Corollary \ref{coro:est-with-bounded-cov}, we can complete the proof of Lemma \ref{lem:rough-estimate}.

\begin{proof}[Proof of Lemma \ref{lem:rough-estimate}]
By Corollary \ref{coro:est-with-bounded-cov} and Claim \ref{claim:cov-upper-bound}, we get that with $1 - (\eps/n)^{10}$ probability, the estimate vector $\wh{v}$ computed by the algorithm satisfies
\[
\norm{ \frac{\sqrt{\phi}}{1 - \phi} \wh{v} - \E_{\pi \sim M(\phi, \pi^*)} [\trunc(v_{\pi}) ] } \leq O\left( \frac{\sqrt{\eps}}{1 - \phi} \right) \,.
\]
Let $\wh{\pi}$ be the permutation computed by the algorithm, i.e. the permutation that sorts the entries of $\wh{v}$.  Claim \ref{claim:sorting} immediately gives
\[
\norm{v_{\wh{\pi}} - v_{\pi^*}} \leq O\left( \frac{\sqrt{\eps}}{1 - \phi} \right)
\]
which completes the proof.
\end{proof}

\subsection{Analysis of Iterative Refinement Steps}\label{sec:iterative-refinement}

Now we analyze the iterative refinement steps.  Note that {\sc Iterative Refinement} relies on the robust estimation algorithm in Theorem \ref{thm:est-with-estimated-cov} which is for estimating the mean of a distribution with near-identity covariance.  If we had $\wh{\pi} = \pi^*$, then indeed the distributions of the vectors
\[
\frac{(1 - \phi)}{\sqrt{\phi}} \trunc\left(v_{\pi}^{\leftarrow \wh{\pi}}\right) ,  \frac{(1 - \phi)}{\sqrt{\phi}} \trunc\left(v_{\pi}^{\rightarrow \wh{\pi}}\right)
\]
would have near-identity covariance for $\pi \sim M(\phi, \pi^*)$ by using insertion sampling and Claim \ref{claim:almost-exact-expressions}.  However, when $\wh{\pi} \neq \pi^*$, controlling the covariance of the above vectors is more difficult.  In particular, we no longer have the independence structure that insertion sampling gives us.  Instead, what we need to do is control the spectral norm error in the covariance matrix as a function of $\norm{v_{\wh{\pi}} - v_{\pi^*}}$.  The key bound, Lemma \ref{lem:error-in-cov} is stated and proved in the next subsection.  

At a high level, we will translate an upper bound on the distance $\norm{v_{\wh{\pi}} - v_{\pi^*}}$ into a combinatorial property about the error in the covariance matrix.  In particular, we show that the only nonzero entries are within some band around the diagonal.  We then bound the actual entries through a careful analysis of the insertion procedure for generating samples from a Mallows model that roughly uses the independence structure for $v_{\pi}^{\leftarrow \pi^*}$ but quantitatively accounts for the difference between $\wh{\pi}$ and $\pi^*$.

\subsubsection{Error in Estimated Covariance Matrix}

\begin{lemma}\label{lem:error-in-cov}
Let $\wh{\pi}$ and $\pi^*$ be two permutations that both match $\id$ on their first and last $\theta$ elements and such that 
\[
\norm{v_{\wh{\pi}} - v_{\pi^*}} \leq \frac{1}{10(1 - \phi)} \,.
\]
Then we have
\[
\norm{\Cov_{\pi \sim M(\phi, \pi^*)}\left( \frac{1 - \phi}{\sqrt{\phi}} \trunc\left(v_{\pi}^{\leftarrow \wh{\pi}} \right) \right) - I}_{\op} \leq \eps + O(\norm{v_{\wh{\pi}} - v_{\pi^*}} ( 1 - \phi) ) \,.
\]
and the same holds for $v_{\pi}^{\rightarrow \wh{\pi}}$.
\end{lemma}

Note that to prove Lemma \ref{lem:error-in-cov}, WLOG we may assume $\pi^* = \id$.  We will use this assumption throughout the remainder of this subsection.  Recall Definition \ref{def:difference-set}.  We have that for any permutation $\pi$ and element $t$,
\begin{equation}\label{eq:perturbation}
v_{\pi}^{\leftarrow \wh{\pi}}[t] = v_{\pi}^{\leftarrow }[t] - \sum_{x, (x,t) \in S(\id, \wh{\pi})} 1_{\pi^{-1}(x) > \pi^{-1}(t)} + \sum_{y, (t,y) \in S(\id, \wh{\pi})} 1_{\pi^{-1}(y) > \pi^{-1}(t)} \,.
\end{equation}
Note that when computing say $\Cov\left( v_{\pi}^{\leftarrow \wh{\pi}}[t] , v_{\pi}^{\leftarrow \wh{\pi}}[t'] \right)$ , the leading term of $ v_{\pi}^{\leftarrow }[t] v_{\pi}^{\leftarrow }[t']$ has covariance $0$ because of insertion sampling.  Thus to compute the covariance of the distribution of $v_{\pi}^{\leftarrow \wh{\pi}}$, it suffices to compute two types of quantities listed below:
\begin{itemize}
    \item $\Cov(  1_{\pi^{-1}(w) < \pi^{-1}(y)},  1_{\pi^{-1}(x) < \pi^{-1}(z)}) $
    \item $ \Cov(1_{\pi^{-1}(x) < \pi^{-1}(z)}, v_{\pi}^{\leftarrow}[y] ) $
\end{itemize}
for some $w,x,y,z$.  First, we make a basic observation that the above covariances are $0$ unless $w,x,y,z$ are in a certain order.
\begin{claim}\label{claim:basic-ordering}
We have
\begin{enumerate}
    \item $\Cov(  1_{\pi^{-1}(w) < \pi^{-1}(y)},  1_{\pi^{-1}(x) < \pi^{-1}(z)}) = 0$ unless two of $w,x,y,z$ are equal or the intervals $[w,y]$ and $[x,z]$ intersect but neither is contained in the other 
    \item $\Cov(1_{\pi^{-1}(x) < \pi^{-1}(z)}, v_{\pi}^{\leftarrow}[y] ) = 0$ unless two of $x,y,z$ are equal or $y$ is contained in the interval $[x,z]$
\end{enumerate}
\end{claim} 
\begin{proof}
For the first part, if the intervals $[w,y]$ and $[x,z]$ are disjoint then it is clear that the two indicator variables are independent by Lemma \ref{lem:restricted-block}.  Now if WLOG $[w,y]$ contains $[x,z]$, then we can use insertion sampling, starting from the middle by inserting the elements in the interval $[x,z]$ first and then inserting the remaining elements in the interval $[w,y]$.  This implies that again the indicator variables are independent.  The second part of the claim also follows immediately from insertion sampling since if $y$ is not in the interval $[x,y]$ then the location where $y$ is inserted is clearly independent of the relative order of $x$ and $z$.   
\end{proof}
For the case when the covariances are nonzero, we prove two key quantitative bounds that are stated below.

\begin{lemma}\label{lem:indicator-covariance}
For any $w,x,y,z\in [n]$ that are all distinct and such that $|w - y|, |x - z| \leq 0.5/(1-\phi)$, we have
\[
\left \lvert \Cov_{\pi \sim M(\phi, \id)}\left(  1_{\pi^{-1}(w) < \pi^{-1}(y)},  1_{\pi^{-1}(x) < \pi^{-1}(z)} \right) \right \rvert  \leq C(1 - \phi) 
\]
for some universal constant $C$.
\end{lemma}

\begin{lemma}\label{lem:estimated-cov}
Let $x,y,z \in [n] $ be elements that are all distinct and such that $|z - x| \leq 1/(1- \phi)$.  Then
\[
\left \lvert \Cov_{\pi \sim M(\phi, \id)} \left( 1_{\pi^{-1}(x) < \pi^{-1}(z)}, v_{\pi}^{\leftarrow}[y] \right)  \right \rvert \leq C 
\]
for some universal constant $C$.
\end{lemma}

Note that the indicator variables $ 1_{\pi^{-1}(x) < \pi^{-1}(z)}$ have constant variance and the variables $v_{\pi}^{\leftarrow}[y]$ have variance $\sim 1/(1 - \phi)^2$.  Thus, the above lemmas can be roughly seen as saying that these random variables have correlation $O(1 - \phi)$ (unless they share some common index).   For some very rough intuition for why Lemma \ref{lem:indicator-covariance} and Lemma \ref{lem:estimated-cov} hold, consider the following.  The value of $v_{\pi}^{\leftarrow}[y]$ only indirectly affects the probability that $x$ and $z$ are inverted because the location where $y$ is inserted can only affect the position of $x$ (relative to the first $y$ elements) by $1$.  Since the locations of elements in $M(\phi, \id)$ vary by roughly $1/(1 - \phi)$, we expect this one position change to affect the relative order of $x$ and $z$ with probability $O(1 - \phi)$.    
We will need a few intermediate results before we can prove these lemmas.  It will be important to recall Definitions \ref{def:front} and \ref{def:back}.  We will consider sampling $\pi \sim M(\phi, \id)$ using insertion sampling from the front.  Note that inserting the first $t$ elements is equivalent to sampling $\pi_{\id:t}$.  The first claim below bounds the expected number of elements among the next $k$ that will be inserted in between two consecutive elements of $\pi_{\id:t}$.

\begin{claim}\label{claim:expected-gaps}
Consider drawing a sample $\pi$ from $M(\phi, \id)$ using the insertion sampling procedure.  Assume that we have sampled $\pi_{\id: t}$ so far for some $t \geq 2$.  Let $x,y$ be two elements in $\pi_{\id:t}$ that are consecutive.  Next consider inserting the next $k$ elements to get $\pi_{\id: t + k}$ where $k \leq \min(t, 1/(1 - \phi))$.  The expected number of elements between $x$ and $y$ in $\pi_{\id: t + k}$ (which is given by $\pi_{\id: t + k}^{-1}(y) - \pi_{\id: t + k}^{-1}(x) - 1$) satisfies
\[
\E[\pi_{\id: t + k}^{-1}(y) - \pi_{\id: t + k}^{-1}(x) - 1 |  \pi_{\id:t}] \leq C k \max(1/t, (1 - \phi))
\]
where $C$ is some (sufficiently large) universal constant.
\end{claim}
\begin{proof}
First we prove the claim assuming that $k \leq 0.1 \min(t, 1/(1 -\phi))$.  We then show how to reduce to this case.  Note that when inserting any element in $s \geq t$ into the permutation $\pi_{\id: s - 1}$, the probability that it lands in any given location is at most 
\[
\frac{1}{1  +\phi + \dots + \phi^{s - 1}} = \frac{1 - \phi}{1 - \phi^s} \leq 2 \max ( 1/s, (1 - \phi)) \leq 2 \max ( 1/t, (1 - \phi)) \,.
\]
The probability that there are exactly $c$ elements between $x$ and $y$ in $\pi_{\id: t + k}$  (where clearly we must have $0 \leq c \leq k$) is at most
\[
\binom{k}{c } c! (2\max ( 1/t, (1 - \phi)))^{c} \leq (2k\max ( 1/t, (1 - \phi)))^{c} 
\]
because there are $\binom{k}{c}$ ways to choose which $c$ elements among $\{t + 1, \dots , t + k \}$ get inserted between $x$ and $y$ and $c!$ ways to choose their relative order.  Once their relative order is fixed, the locations of each of their insertions is also fixed.  Define $\theta = 2k\max ( 1/t, (1 - \phi))$.  We conclude that 
\begin{equation}\label{eq:num-insertions}
\E[\pi_{\id: t + k}^{-1}(y) - \pi_{\id: t + k}^{-1}(x) - 1 |  \pi_{\id:t}] \leq \theta + 2 \theta^2 + 3\theta^3 + \dots \leq \frac{ \theta}{(1 - \theta)^2} \leq  4 k \max(1/t, (1 - \phi))  
\end{equation}
where the last step uses that $k \leq 0.1 \min(t, 1/(1 -\phi))$.  Now we show how to weaken this assumption to $k \leq \min(t, 1/(1 -\phi))$.  We may assume that $k \geq 0.1 \min(t, 1/(1 -\phi))$.  Next, we can imagine inserting $10$ separate blocks of $0.1k$ elements i.e. we first go from $\pi_{\id:t}$ to $\pi_{\id: t + 0.1k}$, then to $\pi_{\id: t + 0.2k}$ and so on.  Note that $0.1 k \leq 0.1 \min(t, 1/(1 -\phi))$ so we can apply (\ref{eq:num-insertions}) repeatedly for each set of $0.1k$ elements that we insert.  Specifically, for all $1 \leq j \leq 10$
\begin{align*}
&\E[\pi_{\id: t + 0.1jk}^{-1}(y) - \pi_{\id: t + 0.1jk}^{-1}(x)  |  \pi_{\id:t + 0.1(j-1)k}] \\ & \quad \leq \left(1 + 0.4 k \max(1/t, (1 - \phi) ) \right) \left( \pi_{\id: t + 0.1(j-1)k}^{-1}(y) - \pi_{\id: t + 0.1(j-1)k}^{-1}(x) \right) \,.
\end{align*}
Using the fact that $k \leq \min(t, 1/(1 - \phi))$ and combining the above for $j = 1,2, \dots , 10$, we conclude
\[
\E[\pi_{\id: t + k}^{-1}(y) - \pi_{\id: t + k}^{-1}(x)  - 1 |  \pi_{\id:t}] \leq 1.4^{10} - 1   \leq 30  \leq 300 k \max(1/t, (1 - \phi))
\]
where the last step follows from the assumption that $k \geq 0.1 \min(t, 1/(1 -\phi))$. This completes the proof. 
\end{proof}

Next, we analyze for three elements, say $x < y < z$, how the position of $x$ relative to the first $y$ elements affects the probability that $x$ and $z$ will be inverted in $\pi$.  

\begin{claim}\label{claim:small-change}
Let $x,y,z \in [n]$ be any elements such that 
\begin{enumerate}
    \item $x  < y < z$
    \item  $z -y \leq \min(y, 1/(1 - \phi))$
\end{enumerate}
Consider drawing a sample $\pi \sim M(\phi, \id)$.  Then for any integer $s$ with $1 \leq s < y$, we have
\[
 \left \lvert \E_{\pi \sim M(\phi, \id)}[ 1_{\pi^{-1}(x) > \pi^{-1}(z)}| \pi_{\id:y}^{-1}(x) = s ]  - \E_{\pi' \sim M(\phi, \id)}[ 1_{\pi'^{-1}(x) > \pi'^{-1}(z)}| \pi_{\id:y}^{\prime -1}(x) = s + 1 ]  \right \rvert \leq C \max(1/z, 1 - \phi)  
\]
for some universal constant $C$.
\end{claim}
\begin{proof}
We consider sampling $\pi$ and $\pi'$ via insertion sampling.  We can assume that we have fixed the permutations $\pi_{\id: y}, \pi'_{\id:y}$ such that the location of $x$ in the two permutations is $s$ and $s+1$ respectively.  Now we can couple the randomness in sampling $\pi$ and $\pi'$ when inserting the the elements $y+1, \dots , z - 1$.  In particular, we can ensure that these elements are inserted into the same locations in $\pi$ and $\pi'$.  This way, in the two permutations $\pi_{\id:z-1}$ and $\pi_{\id: z - 1}'$, the locations of all of the elements $y+1, \dots , z - 1$ are the same.  Now let
\[
c = \pi_{\id:z-1}^{\prime -1}(x) - \pi_{\id:z-1}^{-1}(x) \,.
\]
Observe that 
\begin{equation}\label{eq:expected-diff}
\left \lvert \E[1_{\pi^{-1}(x) < \pi^{-1}(z)}| \pi_{\id: z - 1 } ] - \E[1_{\pi'^{-1}(x) < \pi'^{-1}(z)} |  \pi'_{\id: z - 1}] \right \rvert \leq c \cdot  \frac{1 - \phi}{1 - \phi^z} \,.
\end{equation}
This is because when inserting the element $z$ (with coupled randomness), there are $c$ locations that would result in
\[
1_{\pi^{-1}(x) < \pi^{-1}(z)} \neq 1_{\pi'^{-1}(x) < \pi'^{-1}(z)}
\]
and $z$ is inserted into each of these locations with probability at most 
\[
\frac{1}{1 + \phi + \dots + \phi^{z - 1}} = \frac{1 - \phi}{1 - \phi^z} \,.
\]
Next, we can bound the expected value of $c$ using Claim \ref{claim:expected-gaps} because with the coupled randomness, $c$ is the number of elements among $y+1, \dots , z - 1$ that get inserted between locations $s$ and $s+1$ of $\pi_{\id:y}$.  By Claim \ref{claim:expected-gaps}
\begin{equation}\label{eq:gap}
\E[ \pi_{\id: z-1 }^{-1}(x) - \pi_{\id: z-1}^{\prime -1}(x) - 1] \leq C(z - y)\max(1/y, 1 - \phi)  \,.
\end{equation}
This is because once we have coupled the randomness, the difference $\pi_{\id: z-1 }^{-1}(x) - \pi_{\id: z-1}^{\prime -1}(x) - 1$ is equal to the number of elements among $y+1, \dots , z - 1$ inserted between consecutive elements of $\pi_{\id: y }$. Furthermore, once $\pi_{\id: z-1 }^{-1}(x), \pi_{\id: z-1}^{\prime -1}(x)$ are fixed,

Combining (\ref{eq:expected-diff}, \ref{eq:gap}), we deduce
\begin{align*}
\left \lvert \E_{\pi \sim M(\phi, \id)}[1_{\pi^{-1}(x) < \pi^{-1}(z)} | \pi_{\id: y }^{-1}(x) = s] - \E_{\pi' \sim M(\phi, \id)}[1_{\pi'^{ -1}(x) < \pi'^{-1}(z)} |  \pi_{\id: y}^{\prime -1}(x) = s+1 ] \right \rvert  \\ \leq (1 + C (z - y)\max(1/y, 1 - \phi))  \cdot \frac{1 - \phi}{ 1 - \phi^z}  \\ \leq 2(C + 1) \max(1/z, 1 - \phi)  
\end{align*}
where in the last step we used that $z -y \leq \min(y, 1/(1 - \phi))$.  This completes the proof.
\end{proof}

Now we can complete the proofs of Lemma \ref{lem:indicator-covariance} and Lemma \ref{lem:estimated-cov}.

\begin{proof}[Proof of Lemma \ref{lem:indicator-covariance}]
WLOG $w$ is the smallest i.e. $w < x,y,z$.  Also WLOG assume $x < z$.  Note that we can do this because switching the order of $(x,z)$ or the order of $(w,y)$ simply negates the covariance.  Next, by Claim \ref{claim:basic-ordering}, the covariance is $0$ unless we have exactly $w  <x < y < z$.  Also WLOG $z - y \leq x - w$.  Otherwise, we can simply reverse all of the permutations.  Finally, WLOG, we may set $w = 1$ since the elements before $w$ do not matter and we can imagine sampling using the insertion procedure but inserting all of these elements last.

Now consider sampling $\pi \sim M(\phi, \id)$ according to a modified insertion sampling procedure.  We first sample the induced permutation on $\{2, \dots , y-1 \}$.  Then we insert $1$ and then $y$ according to insertion sampling.  Finally, we then insert $y+1, \dots , z$ according to insertion sampling.  The high-level idea of the proof is as follows: to bound the covariance, it suffices to bound the difference
\[
\E[  1_{\pi^{-1}(x) < \pi^{-1}(z)} | 1_{\pi^{-1}(1) < \pi^{-1}(y)}  ] - \E[  1_{\pi^{-1}(x) < \pi^{-1}(z)}] \,.
\]
We break into cases based on the location of $x$ relative to $\{2, \dots , y - 1 \}$.  Note that regardless of whether we condition on $1_{\pi^{-1}(1) < \pi^{-1}(y)}$, the distribution of the induced permutation on $\{2, \dots , y - 1 \}$ is the same as a Mallows model on $y - 2$ elements.  Once we fix the location of $x$ relative to $\{2, \dots , y - 1 \}$, we analyze sampling the locations of $1,y$ conditioned on $1_{\pi^{-1}(1) < \pi^{-1}(y)} $ and also without any conditioning and see how these different sampling procedures affect the location of $x$.  Once this is done, the location of $x$ relative to $\{1,2, \dots , y \}$ is fixed.  The key step is that we can then apply Claim \ref{claim:small-change} to bound the difference between these two different sampling procedures.

First, assume that element $x$ lands in position $i$ relative to $2, \dots , y - 1$.  Then after inserting elements $1,y$, it will land in some position among $\{i, i+1, i+2 \}$.  The probabilities when conditioned on $ \pi^{-1}(1) < \pi^{-1}(y)$ of $x$ landing in position $i,i+1, i+2$ are
\begin{align*}
q_{i \rightarrow i} &= \frac{\phi^i + 2\phi^{i + 1} + \dots + (y-i - 1)\phi^{y - 2}}{1 + 2\phi + \dots + (y-1)\phi^{y-2}} \\  q_{i \rightarrow i+1 } &= \frac{( 1 + \phi + \dots + \phi^{i - 1})(1 + \phi + \dots + \phi^{y - 2  -i })}{1 + 2\phi + \dots + (y-1)\phi^{y-2}}  \\  q_{i \rightarrow i + 2} &= \frac{\phi^{y - 1 - i} + 2\phi^{y - i} + \dots + i\phi^{y-2}}{1 + 2\phi + \dots + (y-1)\phi^{y-2} } \,.
\end{align*}
The probabilities with no conditioning of $x$ landing in position $i,i+1, i+2$ are 
\begin{align*}
r_{i \rightarrow i} &= \frac{(1+ \phi + \dots + \phi^{i-1})( 1 + \phi + \dots + \phi^i) \phi^{y - i - 1}}{( 1 + \dots + \phi^{y - 1})( 1 + \dots + \phi^{y - 2})} \\ r_{i \rightarrow i+1 } &= \frac{( 1 + \phi + \dots + \phi^{i - 1})(1 + \phi + \dots + \phi^{y - 2  -i }) (1 + \phi^{y} ) }{( 1 + \dots + \phi^{y - 1})( 1 + \dots + \phi^{y - 2})} \\  r_{i \rightarrow i + 2} &= \frac{(1 + \phi + \dots + \phi^{y - 2 - i})( 1 + \phi + \dots + \phi^{y - 1 - i})\phi^i}{( 1 + \dots + \phi^{y - 1})( 1 + \dots + \phi^{y - 2})}
\end{align*}
These probabilities can be computed through direct computation (doing casework on the location where elements $1,y$ are inserted).  Note that we have the following inequalities:
\begin{align}\label{eq:prob-bounds1}
\phi^y \frac{i(i+1)}{(y-1)y} \leq &q_{i \rightarrow i}, r_{i \rightarrow i} \leq \frac{1}{\phi^y} \frac{i(i+1)}{(y-1)y} \\ \label{eq:prob-bounds2}
\phi^y \frac{2i(y-1-i)}{(y-1)y } \leq &q_{i \rightarrow i + 1}, r_{i \rightarrow i+1} \leq \frac{1}{\phi^y} \frac{2i(y-1-i)}{(y-1)y } \\ \label{eq:prob-bounds3}
\phi^y \frac{(y - i - 1)(y - i)}{(y-1)y }  \leq &q_{i \rightarrow i + 2}, r_{i \rightarrow i+2} \leq \frac{1}{\phi^y}  \frac{(y - i - 1)(y - i)}{(y-1)y } \,.
\end{align}
These are obtained by simply counting the number of terms in the numerator and denominator of each expression and using a trivial bound on the ratio of the terms.   Once we have inserted all of the elements $1,2, \dots , y$ (fixing the location of $x$ relative to these elements), it remains to consider the insertions of $y+1, \dots , z$ to get the probability of $1_{\pi^{-1}(x) < \pi^{-1}(z)}$.  We can bound the following quantity
\[
D_i = \E[  1_{\pi^{-1}(x) < \pi^{-1}(z)} | 1_{\pi^{-1}(1) < \pi^{-1}(y)}, \pi_{\id: \{2, \dots , y - 1 \}}^{-1}(x) =  i ] - \E[  1_{\pi^{-1}(x) < \pi^{-1}(z)}| \pi_{\id: \{2, \dots , y - 1 \}}^{-1}(x) =  i] \,.
\]  
using Claim \ref{claim:small-change} (note that we can use Claim \ref{claim:small-change} because we assumed $z - y \leq x - w \leq 1/(1 - \phi)$) and the bounds in (\ref{eq:prob-bounds1}, \ref{eq:prob-bounds2}, \ref{eq:prob-bounds3}).  In particular, note that 
\begin{align*}
&\E[  1_{\pi^{-1}(x) < \pi^{-1}(z)} | 1_{\pi^{-1}(1) < \pi^{-1}(y)}, \pi_{\id: \{2, \dots , y - 1 \}}^{-1}(x) =  i ] =  q_{i \rightarrow i}\E[  1_{\pi^{-1}(x) < \pi^{-1}(z)} | \pi_{\id: y}^{-1}(x) =  i ] \\ & \quad + q_{i \rightarrow i+1}\E[  1_{\pi^{-1}(x) < \pi^{-1}(z)} | \pi_{\id: y}^{-1}(x) =  i + 1 ]  + q_{i \rightarrow i + 2}\E[  1_{\pi^{-1}(x) < \pi^{-1}(z)} | \pi_{\id: y}^{-1}(x) =  i + 2 ] 
\end{align*}
and 
\begin{align*}
&\E[  1_{\pi^{-1}(x) < \pi^{-1}(z)} | \pi_{\id: \{2, \dots , y - 1 \}}^{-1}(x) =  i ]  =  r_{i \rightarrow i}\E[  1_{\pi^{-1}(x) < \pi^{-1}(z)} | \pi_{\id: y}^{-1}(x) =  i ] \\ & \quad + r_{i \rightarrow i+1}\E[  1_{\pi^{-1}(x) < \pi^{-1}(z)} | \pi_{\id: y}^{-1}(x) =  i + 1 ]  + r_{i \rightarrow i + 2}\E[  1_{\pi^{-1}(x) < \pi^{-1}(z)} | \pi_{\id: y}^{-1}(x) =  i + 2 ]
\end{align*}
because the variable $1_{\pi^{-1}(x) < \pi^{-1}(z)}$ is independent of the permutation $\pi_{\id:y}$ once we condition of the location of $x$ in  $\pi_{\id:y}$.  Now by Claim \ref{claim:small-change}, we get
\[
D_i \leq 2(1 - \phi^{2y}) C \max(1/z, 1 - \phi)  \leq 4C y(1 - \phi) \max(1/y, 1 - \phi) \leq 4C(1 - \phi)
\]
where the last step uses the fact that we assumed all of $w,x,y,z$ are contained within an interval of length $1/(1 - \phi)$.  Now combining the above over all possibilities for $i = 1,2, \dots , y-2$, we get 
\[
\E[  1_{\pi^{-1}(x) < \pi^{-1}(z)} | 1_{\pi^{-1}(1) < \pi^{-1}(y)}  ] - \E[  1_{\pi^{-1}(x) < \pi^{-1}(z)}] \leq 4C(1 - \phi) 
\]
and from the above, we immediately get the desired inequality.
\end{proof}

\begin{proof}[Proof of Lemma \ref{lem:estimated-cov}]
WLOG, $x < z$ since switching $x$ and $z$ simply negates the covariance.  Note that by Claim \ref{claim:basic-ordering}, the covariance is $0$ unless we have exactly $x < y < z$ so we may restrict to this case.  Also, WLOG, we may assume $y - x \geq z - y$ since otherwise we can reverse everything.  

Note that $v_{\pi}^{\leftarrow}[y] = y - \pi_{\id:y}^{-1}(y) $.  Thus, it will suffices to compute the covariance of $1_{\pi^{-1}(x) < \pi^{-1}(z)}$ with $\pi_{\id:y}^{-1}(y)$.  To prove the desired inequality, we will rely on Claim \ref{claim:small-change}.  First, note that it suffices to bound the difference
\[
\left \lvert \E_{\pi \sim M(\phi, \id)}\left[1_{\pi^{-1}(x) < \pi^{-1}(z)} | \pi_{\id:y}^{-1}(y) = s \right] - \E_{\pi' \sim M(\phi, \id)}\left[1_{\pi^{\prime -1}(x) < \pi^{ \prime -1}(z)} | \pi_{\id:y}^{\prime -1}(y) = t \right] \right \rvert
\]
for any $1 \leq s,t \leq y$.  Now we can imagine drawing the two samples $\pi$ for the first expectation and $\pi'$ for the second expectation according to the insertion model with coupled randomness.  More specifically, we ensure that elements $1,2, \dots , y-1$ are inserted into the same locations in both.  Element $y$ is inserted into location $s$ in $\pi$ and into location $t$ in $\pi'$.  Finally, elements $y+1, \dots , z$ are again inserted into the same locations in both.  Note that with this coupled randomness, we have
\[
\left \lvert \pi_{\id:y}^{-1}(x) - \pi_{\id:y}^{\prime -1}(x) \right \rvert \leq 1\,.
\]
Once we fix the order of the elements $\{1,2, \dots , y \}$ in the insertion model, the event $1_{\pi^{-1}(x) < \pi^{-1}(z)}$ depends only on the position of $x$ with respect to $\{1,2, \dots , y \}$.  Thus, we may apply Claim \ref{claim:small-change} to deduce
\begin{equation}\label{eq:diff-bound}
\left \lvert \E\left[1_{\pi^{-1}(x) < \pi^{-1}(z)} | \pi_{\id:y}^{-1}(y) = s \right] - \E\left[1_{\pi^{\prime -1}(x) < \pi^{ \prime -1}(z)} | \pi_{\id:y}^{\prime -1}(y) = t \right] \right \rvert \leq C \max(1/z, 1 - \phi) \,.
\end{equation}
 
Note that we have an explicit expression for the distribution of $\pi^{-1}_{\id:y}(y)$ from the insertion model.  In particular, the distribution of $v_{\pi}^{\leftarrow}[y] = y - \pi^{-1}_{\id:y}(y)$ for $\pi \sim M(\phi, \id)$ is exactly $\mcl{D}_y$ (recall Definition \ref{def:insertion-distribution}).  Now we have that
\begin{align*}
&\Cov_{\pi \sim M(\phi, \id)} \left( 1_{\pi^{-1}(x) < \pi^{-1}(z)}, v_{\pi}^{\leftarrow}[y] \right)  \\ &= \sum_{s = 0}^{y-1} \left( \E\left[1_{\pi^{-1}(x) < \pi^{-1}(z)} | v_{\pi}^{\leftarrow}[y] = s \right] - \E\left[1_{\pi^{-1}(x) < \pi^{-1}(z)}\right] \right)  \cdot \Pr[v_{\pi}^{\leftarrow}[y] = s ] \cdot s  \\ & \leq C \max(1/z , 1 - \phi) \sum_{s = 0}^{y-1}  \Pr_{\gamma \sim \mcl{D}_y}[\gamma = s ] \cdot s  \\ & \leq  C \max(1/z , 1 - \phi) \min \left( y, \frac{\phi}{1 - \phi}\right) \\ & \leq C \,.
\end{align*}
where we used (\ref{eq:diff-bound}) and then Claim \ref{claim:explicit-means}.  This completes the proof. 
\end{proof}

Lemma \ref{lem:indicator-covariance} and Lemma \ref{lem:estimated-cov} are the key ingredients in the proof of Lemma \ref{lem:error-in-cov}.  We will need one more simple identity before we are ready to finish the proof.

\begin{claim}\label{claim:inversion-bound}
For any permutation $\pi$ on $[n]$, 
\[
\sum_{t = 1}^n (v_{\pi}^{\leftarrow}[t] + v_{\pi}^{\rightarrow}[t])^2 \leq 4\norm{v_{\id} - v_{\pi}}^2 \,.
\]
\end{claim}
\begin{proof}
Consider a fixed element $t$.  Note that  
\begin{align*}
\sum_{(x,t) \in S(\id, \pi)} ( t  - x) + \sum_{(t, y) \in S(\id, \pi)} (y - t) \geq (1 + \dots + v_{\pi}^{\leftarrow}[t]) + (1 + \dots + v_{\pi}^{\rightarrow}[t]) \geq \frac{v_{\pi}^{\leftarrow}[t]^2 + v_{\pi}^{\rightarrow}[t]^2}{2} \\ \geq  \frac{(v_{\pi}^{\leftarrow}[t] + v_{\pi}^{\rightarrow}[t])^2 }{4} \,.
\end{align*}
Summing the above over all $t$, we get
\begin{align*}
\sum_{t = 1}^n (v_{\pi}^{\leftarrow}[t] + v_{\pi}^{\rightarrow}[t])^2 \leq 4 \sum_{t = 1}^n \left( \sum_{(x,t) \in S(\id, \pi)} ( t  - x) + \sum_{(t, y) \in S(\id, \pi)} (y - t) \right) = 8 \sum_{(x,y) \in S(\id, \pi)} ( y - x) \\ = 4\norm{v_{\id} - v_{\pi}}^2 
\end{align*}
where the last step follows from Claim \ref{claim:sum-identity}.
\end{proof}

Now we can complete the proof of Lemma \ref{lem:error-in-cov}.

\begin{proof}[Proof of Lemma \ref{lem:error-in-cov}]
WLOG $\pi^* = \id$.  By Claim \ref{claim:inversion-bound}, we know that for all $t \in [n + 2\theta]$, we have
\begin{equation}\label{eq:inversion-bound}
v_{\wh{\pi}}^{\leftarrow}[t] + v_{\wh{\pi}}^{\rightarrow}[t] = |x| (x,t) \in S(\id, \wh{\pi}) | +  |y| (t,y) \in S(\id, \wh{\pi}) | \leq 2 \norm{ v_{\wh{\pi}} - v_{\id}} \,.
\end{equation}
Now we must bound the following two types of quantities to control the difference in the covariance matrices.
\begin{itemize}
    \item Diagonal Entries: $ \left \lvert \Var_{\pi \sim M(\phi, \id)}( v_{\pi}^{\leftarrow \wh{\pi}}[t] ) - \Var_{\pi \sim M(\phi, \id)}( v_{\pi}^{\leftarrow}[t]) \right \rvert$ 
    \item Off-diagonal Entries: $ \left \lvert \Cov_{\pi \sim M(\phi, \id)}\left( v_{\pi}^{\leftarrow \wh{\pi}}[t] , v_{\pi}^{\leftarrow \wh{\pi}}[t']\right)  \right \rvert$ for $t \neq t'$
\end{itemize}
Furthermore, we may assume $t,t' \geq \theta$ (because of the truncation).  Note that for any $x,t$, we have
\begin{equation}\label{eq:variance}
\left \lvert \Cov( v_{\pi}^{\leftarrow}[t], 1_{\pi^{-1}(x) > \pi^{-1}(t)}) \right \rvert \leq \sqrt{\Var( v_{\pi}^{\leftarrow}[t]) } \leq \frac{2}{(1 - \phi)} 
\end{equation}
by Claim \ref{claim:almost-exact-expressions} and Claim \ref{claim:explicit-variance}.  Thus, using (\ref{eq:perturbation}) and (\ref{eq:inversion-bound}), we have
\begin{equation}\label{eq:diag-bound}
\left \lvert \Var_{\pi \sim M(\phi, \id)}( v_{\pi}^{\leftarrow \wh{\pi}}[t] ) - \Var_{\pi \sim M(\phi, \id)}( v_{\pi}^{\leftarrow}[t]) \right \rvert \leq \frac{4 \norm{ v_{\wh{\pi}} - v_{\id}}}{(1 - \phi)} +  4\norm{ v_{\wh{\pi}} - v_{\id}}^2 \leq \frac{5 \norm{ v_{\wh{\pi}} - v_{\id}} }{1 - \phi} \,.
\end{equation}
Next we consider the off-diagonal entries. We can write
\begin{align*}
&\left \lvert \Cov_{\pi \sim M(\phi, \id)}\left( v_{\pi}^{\leftarrow \wh{\pi}}[t] , v_{\pi}^{\leftarrow \wh{\pi}}[t']\right)  \right \rvert \\ &\leq  \left \lvert \Cov \left(v_{\pi}^{\leftarrow}[t],  \sum_{x, (x,t') \in S(\id, \wh{\pi})} 1_{\pi^{-1}(x) > \pi^{-1}(t')} -  \sum_{y, (t',y) \in S(\id, \wh{\pi})} 1_{\pi^{-1}(t') > \pi^{-1}(y)}  \right) \right \rvert \\ & \quad + \left \lvert \Cov \left(v_{\pi}^{\leftarrow}[t'],  \sum_{x, (x,t) \in S(\id, \wh{\pi})} 1_{\pi^{-1}(x) > \pi^{-1}(t)} -  \sum_{y, (t,y) \in S(\id, \wh{\pi})} 1_{\pi^{-1}(t) > \pi^{-1}(y)}  \right) \right \rvert \\ & \quad + \Bigg \lvert \Cov \Bigg(\sum_{x, (x,t) \in S(\id, \wh{\pi})} 1_{\pi^{-1}(x) > \pi^{-1}(t)} -  \sum_{y, (t,y) \in S(\id, \wh{\pi})} 1_{\pi^{-1}(t) > \pi^{-1}(y)} , \\ &\quad\quad\quad\quad  \sum_{x, (x,t') \in S(\id, \wh{\pi})} 1_{\pi^{-1}(x) > \pi^{-1}(t')} -  \sum_{y, (t',y) \in S(\id, \wh{\pi})} 1_{\pi^{-1}(t') > \pi^{-1}(y)}  \Bigg) \Bigg \rvert \,. 
\end{align*}
Let the three expressions on the RHS be $S_1, S_2, S_3$.  We can bound $S_1, S_2$ using Lemma \ref{lem:estimated-cov} and bound $S_3$ using Lemma \ref{lem:indicator-covariance}.  To see why we can apply these lemmas, note that for any $(x,y) \in S(\id, \wh{\pi})$, we must have that 
\[
|x - y| \leq 2 \norm{v_{\id} - v_{\wh{\pi}} } \leq 1/(5(1 - \phi)) \,.
\]
Thus, examining the expression (\ref{eq:perturbation}) and by Claim \ref{claim:basic-ordering}, the only way for $v_{\pi}^{\leftarrow \wh{\pi}}[t]$ and $v_{\pi}^{\leftarrow \wh{\pi}}[t']$ to not be independent is if $|t - t'| \leq 4 \norm{v_{\id} - v_{\wh{\pi}} }$.   We now have
\[
S_1 \leq \frac{2}{1 - \phi} + 2C\norm{v_{\wh{\pi}} - v_{\id}}  \leq \frac{C}{1 - \phi} 
\]
for some universal constant $C$.  The above is because there is at most one indicator variable in the second expression that involves $t$ and we can bound the covariance of this one with $v_{\pi}^{\leftarrow}[t]$ using (\ref{eq:variance}).  We can then apply Lemma \ref{lem:estimated-cov} and use (\ref{eq:inversion-bound}) for the remaining terms.  Similarly, we get the same bound for $S_2$.  Finally, we get that
\[
S_3 \leq  6 \norm{v_{\id} - v_{\wh{\pi}}} + 4C \norm{v_{\id} - v_{\wh{\pi}}}^2(1 - \phi) \leq \frac{C}{1 - \phi}
\]
where the first term comes from counting the number of pairs of indicator variables that share some index, using (\ref{eq:inversion-bound}) and a trivial bound on the covariance and the second comes from using (\ref{eq:inversion-bound}) and Lemma \ref{lem:indicator-covariance}.  Overall, we conclude that for all $t,t'$, 
\begin{equation}\label{eq:off-diagonal}
\left \lvert \Cov_{\pi \sim M(\phi, \id)}\left( v_{\pi}^{\leftarrow \wh{\pi}}[t] , v_{\pi}^{\leftarrow \wh{\pi}}[t']\right)  \right \rvert  \leq \frac{3C}{1 -\phi}    
\end{equation}
where $C$ is some universal constant.  Now consider the matrix
\[
D = \Cov\left( \trunc\left(v_{\pi}^{\leftarrow \wh{\pi}} \right)\right) - \Cov\left( \trunc\left(v_{\pi}^{\leftarrow } \right)\right) \,.
\]
Recall that we argued that for all $t,t'$ with $|t - t'| \leq 4\norm{v_{\id} - v_{\wh{\pi}}}$, the entry $D_{t,t'}$ is $0$.  The spectral norm of $D$ is upper bounded by the maximum of the sum of the absolute values of the entries in any row.  Using (\ref{eq:diag-bound}) and (\ref{eq:off-diagonal}), we conclude
\[
\norm{D}_{\op} \leq \frac{3C}{1 - \phi} \cdot 8\norm{v_{\id} - v_{\wh{\pi}}} + \frac{5\norm{v_{\id} - v_{\wh{\pi}}}}{(1 - \phi)} \leq \frac{C'\norm{v_{\id} - v_{\wh{\pi}}} }{1 - \phi}
\]
for some universal constant $C'$.  Combining the above with Claim \ref{claim:almost-exact-expressions}, which upper bounds the difference 
\[
\norm{ \frac{1 - \phi}{\sqrt{\phi}}\Cov\left( \trunc\left(v_{\pi}^{\leftarrow }\right) \right) - I}_{\op}
\]
we immediately get the desired inequality.  The exact same argument works for $v_{\pi}^{\rightarrow \wh{\pi}}$.
\end{proof}

\subsubsection{Stability Bounds for Adjustment Vectors}

To be able to use Theorem \ref{thm:est-with-estimated-cov}, we need to turn the bound in Lemma \ref{lem:error-in-cov} into a stability bound.  Fortunately, this will follow by combining the tail decay properties in Lemma \ref{lem:Gaussian-tails} with the concentration given by Corollary \ref{coro:high-dim-stability}.  The stability bound that we prove is stated below.

\begin{corollary}\label{coro:perm-stability}
Let $\wh{\pi}$ and $\pi^*$ be permutations on $n + 2\theta$ elements such that 
\begin{itemize}
    \item They both match $\id$ on their first and last $\theta$ elements
    \item 
    \[
    \norm{v_{\wh{\pi}} - v_{\pi^*}} \leq \frac{1}{10(1 - \phi)}
    \]
\end{itemize}
For $\pi \sim M(\phi, \pi^*)$, let $\mcl{D}$ be the distribution of the vector 
\[
\frac{(1 - \phi)}{\sqrt{\phi}} \trunc\left( v_{\pi}^{\leftarrow \wh{\pi}}\right) \,.
\]
A set of samples $S$ drawn i.i.d from $\mcl{D}$ with $|S| \geq (n/\eps)^2\poly(\log(n/\eps))$ is stable with parameters
\[
\left(\eps , C \left(\sqrt{\eps \norm{v_{\pi^*} - v_{\wh{\pi}}} (1 - \phi) } + \eps \log 1/\eps\right) \right)
\]
with probability at least $1 - 2^{-9 n \log n /\eps}$ where $C$ is some universal constant.  A similar bound holds for $v_{\pi}^{\rightarrow \wh{\pi}}$.
\end{corollary}
\begin{proof}
WLOG $\pi^* = \id$.  Next, observe that for any permutation $\pi$, we have
\begin{equation}\label{eq:coupling}
\norm{v_{\pi}^{\leftarrow \wh{\pi}} - v_{\pi}^{\leftarrow} } \leq \norm{v_{\wh{\pi}}^{\leftarrow} + v_{\wh{\pi}}^{\rightarrow} } \leq 2\norm{v_{\id} - v_{\wh{\pi}}} \leq 2/(1 - \phi) \,.
\end{equation}
The first step is because on the LHS of the above, the $i$\ts{th} entry is at most the number of elements $j$ such that $i$ and $j$ are in a different order in $\wh{\pi}$ and $\id$.  We then used Claim \ref{claim:inversion-bound} and the assumption in the statement.  Combining with Lemma \ref{lem:Gaussian-tails}, we get that for any unit vector $v$, 
\[
\Pr_{\pi \sim M(\phi, \id)}\left[ \left \lvert v \cdot v_{\pi}^{\leftarrow \wh{\pi}} - \E[v \cdot v_{\pi}^{\leftarrow\wh{\pi}}] \right \rvert \geq  \frac{t + 4}{1-  \phi}  \right] \leq 4e^{-0.2t} \,.  
\]
Now combining Corollary \ref{coro:high-dim-stability} and Lemma \ref{lem:error-in-cov}, we get that as long as $C$ is sufficiently large, with probability $1 - 2^{-10 n \log n /\eps}$, the set of samples $S$ is stable with parameters
\[
\left(\eps , C \left(\sqrt{\eps \norm{v_{\pi^*} - v_{\wh{\pi}}} (1 - \phi) } + \eps \log 1/\eps\right) \right) \,,
\]
as desired.
\end{proof}

\subsubsection{Finishing the Analysis of {\sc Iterative Refinement}}

We can now complete the proof of Lemma \ref{lem:refinement} by combining Corollary \ref{coro:perm-stability} with Theorem \ref{thm:est-with-estimated-cov}.

\begin{proof}[Proof of Lemma \ref{lem:refinement}]
Note that we can union bound Corollary \ref{coro:perm-stability} over all permutations $\wh{\pi}$ for which it is applicable.   Thus, with $1 - 2^{-8 n \log n /\eps}$ probability, for all $\wh{\pi}$ satisfying the hypotheses of Corollary \ref{coro:perm-stability}, the set of samples $S$ of
\[
\frac{(1 - \phi)}{\sqrt{\phi}} \trunc\left( v_{\pi}^{\leftarrow \wh{\pi}}\right)    
\]
is stable with parameters
\[
\left( 3 \eps , C \left(\sqrt{\eps \norm{v_{\pi^*} - v_{\wh{\pi}}} (1 - \phi) } + \eps \log 1/\eps\right) \right)
\]
where $C$ is a sufficiently large universal constant and the same statement holds for $v_{\pi}^{\rightarrow \wh{\pi}}$.  Now by Theorem \ref{thm:est-with-estimated-cov}, the estimates $\wh{v^{\leftarrow }}, \wh{v^{\rightarrow}}$ computed by {\sc Iterative Refinement} satisfy 
\begin{align*}
\norm{ \wh{v^{\leftarrow }} - \frac{(1 - \phi)}{\sqrt{\phi}} \E_{\pi \sim M(\phi, \pi^*)}\left[\trunc\left( v_{\pi}^{\leftarrow \wh{\pi}}\right)   \right] } &\leq C \left(\sqrt{\eps \norm{v_{\pi^*} - v_{\wh{\pi}}} (1 - \phi) } + \eps \log 1/\eps\right)\\ 
\norm{ \wh{v^{\rightarrow }} - \frac{(1 - \phi)}{\sqrt{\phi}} \E_{\pi \sim M(\phi, \pi^*)}\left[\trunc\left( v_{\pi}^{\rightarrow \wh{\pi}}\right)   \right] } & \leq C \left(\sqrt{\eps \norm{v_{\pi^*} - v_{\wh{\pi}}} (1 - \phi) } + \eps \log 1/\eps\right) \,.
\end{align*}
Recall Claim \ref{claim:split-vector}.  Since for any permutation $\pi$ we have $v_{\pi}  = v_{\wh{\pi}} +     v_{\pi}^{\rightarrow \wh{\pi}} - v_{\pi}^{\leftarrow \wh{\pi}}$, the estimate $\wh{v}$ computed by {\sc Iterative Refinement} must satisfy
\[
\norm{ \wh{v} - \frac{(1 - \phi)}{\sqrt{\phi}} \E_{\pi \sim M(\phi, \pi^*)}\left[\trunc(v_{\pi})  \right] }  \leq 2C \left(\sqrt{\eps \norm{v_{\pi^*} - v_{\wh{\pi}}} (1 - \phi) } + \eps \log 1/\eps\right) \,.
\]
Now Claim \ref{claim:sorting} implies that the output $\wt{\pi}$ computed by {\sc Iterative Refinement} satisfies
\[
\norm{v_{\wt{\pi}} - v_{\pi^*}} \leq 6C \left( \sqrt{\frac{\eps  \norm{v_{\wh{\pi}} - v_{\pi^*} }}{1 - \phi} } + \frac{\eps \log 1/\eps}{1 - \phi}  \right)
\]
(recall that we assume $\phi \geq 1/2$ so we can clear the $\sqrt{\phi}$ from the denominator at a cost to the constant in front) which completes the proof.
\end{proof}

\subsection{Putting Everything Together}

Now we can complete the proof of Theorem \ref{thm:analysis-main} by combining Lemma \ref{lem:rough-estimate} and Lemma \ref{lem:refinement}. 

\begin{proof}[Proof of Theorem \ref{thm:analysis-main}]
By Lemma \ref{lem:rough-estimate}, with $ 1 - (\eps/n)^{10}$ probability, the first rough estimate $\wh{\pi}$ computed by {\sc Full Algorithm} satisfies
\[
\norm{v_{\wh{\pi}} - v_{\pi^*} } \leq  \frac{C \sqrt{\eps}}{1 - \phi} \leq \frac{1}{10 ( 1 - \phi)}
\]
where the second inequality uses that $\eps$ is sufficiently small.  Now with probability $1 - (\eps/n)^{10}$, the hypotheses of Lemma \ref{lem:refinement} are satisfied.  Whenever, 
\[
\norm{v_{\wh{\pi}} - v_{\pi^*} } \geq \frac{(10C)^2 \eps \log 1/\eps}{1 - \phi} \,,
\]
Lemma \ref{lem:refinement} implies that the next refined estimate $\wt{\pi}$ satisfies
\[
\norm{v_{\wt{\pi}} - v_{\pi^*} }  \leq 0.5 \norm{v_{\wh{\pi}} - v_{\pi^*} }
\]
i.e. the error is halved.  Since we run $\lceil 10 \log 1/\eps \rceil$ iterations of {\sc Iterative Refinement}, we must eventually get to some estimate $\wt{\pi}$ such that 
\[
\norm{v_{\wt{\pi}} - v_{\pi^*} } \leq  \frac{(10C)^2 \eps \log 1/\eps}{1 - \phi} \,.
\]
By the guarantees of Lemma \ref{lem:refinement}, all future refinements must satisfy the above inequality as well.  In particular, the output must satisfy the above inequality and we are done.
\end{proof}

\section{Final Steps}\label{sec:hypothesis-test}

\subsection{Hypothesis Testing}
Now we can put everything together and prove part $(b)$ of Theorem \ref{thm:main}.  Recall that the algorithm in Section \ref{sec:main-alg} requires guessing over $\phi = 0, \eps/n^2, 2\eps/n^2, \dots $ and learning a permutation for each guess of $\phi$.  This, we will end up with a list of $\poly(n/\eps)$ candidate Mallows models and we must run a hypothesis test to select one that is actually close to the true distribution.  Fortunately, this can be done with a standard robust hypothesis testing routine.

\begin{lemma}[See e.g. \cite{kane2020robust, liu2021settling,bakshi2020robustly}  ] \label{lem:hypothesis-test}
Let $\mcl{F}$ be a family of distributions on some domain $\mcl{X}$ with explicitly computable density functions that can be efficiently sampled from.  Let $\mcl{D}$ be an unknown distribution in $\mcl{F}$.  Let $m$ be a parameter.  Let $X_1, \dots , X_s$ be an $\eps$-corrupted sample from $\mcl{D}$ with $s \geq \Omega( \log (m|\mcl{F}|) /\eps^2)$.  Let $H_1, \dots , H_m$ be distributions in $\mcl{F}$ given to us by an adversary with the promise that 
\[
\min_{i} D_{\TV}( H_i, \mcl{D}) \leq \eps \,.
\]
Then there exists an algorithm that runs in $\poly(s/\eps)$ time and with probability $1 - m^{-10}$ outputs an $i \in [m]$ such that 
\[
D_{\TV}( H_i, \mcl{D}) \leq O(\eps) \,.
\]
\end{lemma}

We now complete the proof of part $(b)$ in Theorem \ref{thm:main}.
\begin{proof}[Proof of part $(b)$ in Theorem \ref{thm:main}]
Let $ k = \lfloor n^2/\eps\rfloor $.  Recall that by Theorem \ref{thm:analysis-main}, if we apply the reductions in Section \ref{sec:reductions} and then run Algorithm \ref{alg:full} for all of the guesses for $\phi$, we will obtain permutations $\wt{\pi_0}, \wt{\pi_1}, \dots , \wt{\pi_k}$ such that with probability at least $1 - (\eps/n)^{9}$, there is some $i$ such that
$|i \cdot \eps/n^2 - \phi| \leq \eps/n^2$ and 
\[
\norm{v_{\wt{\pi_i}} - v_{\pi^*}} \leq O\left( \frac{\eps \log 1/\eps}{1 - i \cdot \eps/n^2}\right) \,.
\]
By Theorem \ref{thm:TV-characterization} and Claim \ref{claim:small-TV}, this implies
\[
D_{\TV}( M( i \cdot \eps/n^2, \wt{\pi_i}), M(\phi, \pi^*) ) \leq O(\eps \log 1/\eps) \,. 
\]
Now all of the hypotheses $M( i \cdot \eps/n^2, \wt{\pi_i})$ come from a fixed family of distributions $\mcl{F}$ of size $(k+1) \cdot n!$ (since there are  $k+1$ possible choices for the scaling parameter and the $n!$ choices for the base permutation).  Also note that the insertion procedure both gives us an efficient sampling procedure and an explicit formula for the density functions of all of these distributions.  Thus, by Lemma \ref{lem:hypothesis-test}, with probability at least $1 - (\eps/n)^8$, we can guarantee to output a hypothesis $M(\wt{\phi}, \wt{\pi^*})$ such that 
\[
D_{\TV}( M(\wt{\phi}, \wt{\pi^*}) , M(\phi, \pi^*) ) \leq O(\eps \log 1/\eps)
\]
and we are done.
\end{proof}

\subsection{Parameter Learning}

Here, we prove part $(a)$ in Theorem~\ref{thm:main}, namely that we are also guaranteed to learn the true parameters to nearly optimal accuracy.  Note that in the proof of part $(b)$ in Theorem \ref{thm:main}, we already argued that among the candidates found by Algorithm \ref{alg:full}, one of them has parameters close to the true Mallows model.  It remains to argue that any candidate that is accepted by the hypothesis test must also have parameters close to the true parameters.  Note that Theorem \ref{thm:TV-characterization} and Theorem \ref{thm:TV-lowerbound} already imply that if two Mallows models with the same scaling parameter are close in TV distance, then they must have close base permutations.  We now show that if two Mallows models are close in TV distance, then they must have close scaling parameters.

\begin{claim}\label{claim:scaling-param-to-TV}
Let $\pi^*$ be a permutation on $n$ elements and let $0 < \phi, \phi' < 1$ be parameters.  Then 
\[
D_{\TV}( M(\phi, \pi^*), M(\phi', \pi^*)) \geq \Omega( \min(1, \sqrt{n} |\phi - \phi'|))
\]
\end{claim}
\begin{proof}
WLOG $\pi^* = \id$ and $\phi \geq \phi'$.  Now consider the pairs of elements $(1,2), (3,4), \dots , (n-1,n)$.  For a permutation $\pi$ on $[n]$, we consider a binary vector $z_{\pi} \in \{0,1 \}^{n/2}$ where each entry of $z_{\pi}$ corresponds to whether one of these pairs is inverted in $\pi$ or not.  By Lemma \ref{lem:restricted-block}, the distribution of $z_{\pi}$ for $\pi \sim M(\phi, \pi^*)$ is 
\[
\mcl{D} = \text{Bernoulli}(\phi/(1 + \phi))^{n/2}
\]
while for $\pi \sim M(\phi', \pi^*)$, the distribution is 
\[
\mcl{D}' = \text{Bernoulli}(\phi'/(1 + \phi'))^{n/2} \,.
\]
If we consider the variable $s_{\pi}$, given by the sum of the entries of $z_{\pi}$, the means for $ M(\phi, \pi^*)$ and $ M(\phi', \pi^*)$ differ by
\[
\frac{n}{2} \left( \frac{\phi}{1 + \phi} - \frac{\phi'}{1 + \phi'}\right) \geq \frac{n |\phi - \phi'|}{8}
\]
and the distributions are $\text{Binomial}(n/2 , \phi/(1+ \phi) ) $ and $\text{Binomial}(n/2 , \phi'/(1+ \phi') ) $ respectively.  The desired TV distance lower bound now follows from standard inequalities (see \cite{adell2006exact}).
\end{proof}
\begin{remark}
The above inequality is essentially tight in the regime where $\phi$ is bounded away from $1$ by some constant and $n \rightarrow \infty$.
\end{remark}

We can now prove the parameter learning guarantees in Theorem \ref{thm:main}.

\begin{proof}[Proof of part $(a)$ in Theorem \ref{thm:main}]
The first statement follows immediately from Claim \ref{claim:scaling-param-to-TV}.  Next, note that we must have
\[
D_{\TV}(M(\wt{\phi}, \wt{\pi^*}), M(\phi, \wt{\pi^*}) )  \leq D_{\TV}(M(\wt{\phi}, \wt{\pi^*}), M(\phi, \pi^*) ) \,. 
\]
Combining this with the guarantees of Theorem \ref{thm:main}, we deduce
\[
D_{\TV}(M(\phi, \wt{\pi^*}), M(\phi, \pi^*) ) \leq O(\eps \log 1/\eps) \,.
\]
Finally, applying Theorem \ref{thm:TV-lowerbound} completes the proof.
\end{proof}

\bibliographystyle{alpha}
\bibliography{bibliography}

\appendix

\section{Further Discussion of Results}
\subsection{Why $L^2$ Distance is Optimal}\label{sec:whyL2}
%

As we discussed earlier, bounds on the $L^2$ distance between permutations are the best that we can hope for. To make this precise:
\begin{theorem}\label{thm:optimality}
Let $d$ be any distance function between permutations. Consider trying to learn a Mallwos model $M(\phi^*, \pi^*)$ in the $\epsilon$-strong contamination model.  If it is information-theoretically possible to obtain an estimate $\pi$ for the central ranking that satisfies
$d(\pi, \pi^*) \leq \Delta$.
then 
every permutation $\pi'$ at $L^2$ distance $\eps/(1 - \phi^*)$ from $\pi^*$ satisfies $d(\pi', \pi^*) \leq O(\Delta)$. 
\end{theorem}
\begin{remark}
In particular, the guarantees that are possible in any other distance metric are obtainable (up to logarithmic factors) by simply using $L^2$ distance as a proxy and applying our results.
\end{remark}
\begin{proof}
Consider any two permutations $\pi', \pi^*$ such that $\norm{v_{\pi'} - v_{\pi^*}} \leq  \eps/(1 - \phi^*)$.  Then by Theorem~\ref{thm:TV-informal}, there is an absolute constant $C$ such that 
\[
d_{\TV}\left( M(\phi^*, \pi') M(\phi^*, \pi^*)\right) \leq C\eps \,.
\]
Thus, we can find a sequence of $C$ permutations $\pi_0 = \pi', \pi_1, \dots , \pi_C, \pi_{C+1} = \pi^*$ such that for all terms $j$,
\[
d_{\TV}\left( M(\phi^*, \pi_j) M(\phi^*, \pi_{j+1})\right) \leq \eps \,.
\]
Now, with $\eps$-corruptions, it is information-theoretically impossible to distinguish between whether the true permutation were $\pi_j$ or $\pi_{j + 1}$.  Thus, we must have $d(\pi_j, \pi_{j+1}) \leq 2\Delta$ in order for the stated guarantee to be information-theoretically possible.  Summing over all $j$, we conclude 
\[
d(\pi', \pi^*) \leq 2C \Delta 
\]
as desired.

\end{proof}

\subsection{The Power of Global Information}\label{sec:discussion}

Taking a step back, we ask: {\em Can robust voting rules work by elliciting pairwise preferences alone?} In particular, we could imagine a model where our samples still come from a corrupted Mallows model but where for each voter we only get information about which pair of candidates they prefer, for a randomly chosen pair. In this setting, we show strong lower bounds on the accuracy that can be achieved by any algorithm:

\begin{theorem}\label{thm:impossibility}
Let $K$ be a constant.  There are two Mallows models $M(\phi, \pi^*)$ and $M(\phi, \sigma^*)$ on $[n]$ such that 
\[
D_{\TV}(M(\phi, \pi^*) , M(\phi, \sigma^*)) \geq    \Omega \left( \frac{\eps \sqrt{n}}{K} \right) 
\]
but the induced distribution of $M(\phi, \pi^*)$ and $M(\phi, \sigma^*)$ on any $K$ elements of $[n]$ are $\eps$-close in TV distance.  
\end{theorem}

Theorem \ref{thm:impossibility} immediately implies that any algorithm that observes each voters' preference on only a constant sized set of alternatives cannot achieve accuracy better than $\Omega(\eps \sqrt{n})$.  Furthermore, Theorem~\ref{thm:impossibility} suggests that ``local" algorithms, based on things like learning the relative order of pairs of elements without using global information cannot go beyond $\Omega(\eps \sqrt{n})$. Combined with our positive results, this means:
\begin{quote}
{\em Global ranking information is necessary for achieving accuracy that does not degrade with the number of alternatives.}
\end{quote}

\begin{proof}[Proof of Theorem \ref{thm:impossibility}]
Set $\phi = 1 - \eps/(2K)$.  Let $\pi^* = \id$ and let $\sigma^* = (2,1,4,3, \dots , )$ i.e. we swap adjacent pairs in the identity permutation.  Consider any subset $\{i_1, \dots , i_K \} \subset [n]$.  Observe that for any locations $j_1, \dots , j_K \in [n]$, we have
\[
\phi^K\leq \frac{\Pr_{\pi \sim M(\phi, \pi^*)}[ \sigma^{-1}(i_1) = j_1, \dots , \sigma^{-1}(i_K) = j_K ]}{\Pr_{\sigma \sim M(\phi, \sigma^*)}[ \sigma^{-1}(i_1) = j_1, \dots , \sigma^{-1}(i_K) = j_K ]  } \leq \phi^{-K} \,.
\]
Thus, for any $K$ elements, the TV distance between the distributions $M(\phi, \pi^*)$ and $M(\phi, \sigma^*)$ when restricted to the induced permutation on those $K$ elements is at most  $(1 - \phi^{K}) < \eps$.  It remains to note that $\norm{v_{\pi^*} - v_{\sigma^*}} = \sqrt{n}$ and by Theorem \ref{thm:TV-lowerbound}, as long as $ \sqrt{n} \leq K/\eps$ then we have
\[
D_{\TV}\left( M(\phi, \pi^*), M(\phi, \sigma^*)\right) \geq \Omega\left( \frac{\eps \sqrt{n}}{K}\right)
\]
as desired.
\end{proof}

\section{TV Distance Characterization-Lower Bound Proof}\label{appendix:TV-lowerbound}

In this section, we prove Theorem \ref{thm:TV-lowerbound}.  The proof of Theorem \ref{thm:TV-lowerbound} will be more difficult than the proof of the corresponding upper bound in Theorem \ref{thm:TV-characterization}.  This is because we cannot pass through KL divergence (since we would need some form of a reverse Pinsker inequality).  Recall that (\ref{eq:ratio}) says
\[
\frac{\Pr_{M(\phi, \pi^*)}[\pi]}{\Pr_{M(\phi, \sigma^*)}[\pi]} = \phi^{\Delta_{\pi^*,\sigma^*}(\pi)}  \,.
\]
While controlling the expectation of $\Delta_{\pi^*,\sigma^*}(\pi)$ for $\pi \sim M(\phi, \pi^*)$ sufficed for obtaining an upper bound on TV distance, we will need more precise bounds on the distribution of $\Delta_{\pi^*,\sigma^*}(\pi)$ to obtain the corresponding lower bound.  In particular, we will prove a strong concentration inequality that will essentially allow us to reverse the direction of Pinsker's inequality.

\subsection{Tail Bound via Martingale Concentration}

Note that to compute $\Delta_{\pi^*, \sigma^*}(\pi)$, it suffices to count over all pairs $(x,y) \in S(\pi^*, \sigma^*)$, which of them occur in the order given by $\pi^*$ and which of them occur in the order given by $\sigma^*$.  Now we explain how we obtain a concentration inequality for the distribution of $\Delta_{\pi^*, \sigma^*}(\pi)$.  At a high-level, we will draw a sample from $M(\phi, \pi^*)$ using the insertion sampling process and analyze how each insertion affects the final value of $\Delta_{\pi^*, \sigma^*}(\pi)$.  We will then combine our bounds for each insertion using Azuma's inequality for martingales to obtain the final concentration inequality.  
Recall that drawing a sample $\pi$ from a Mallows model $M(\phi, \pi^*)$ using the insertion sampling procedure involves iteratively sampling the permutations $\pi_{\pi^*:1}. \dots , \pi_{\pi^*:n}$.  We will analyze the Doob martingale.  Formally, 
\begin{definition}
Define the random variables 
\[
X_t = \E[ \Delta_{\pi^*, \sigma^*}(\pi) | \pi_{\pi^*:t}]
\]
for $0 \leq t \leq n$ where the randomness in the expectation is over the remaining $ n - t$ insertions.
\end{definition}

It is clear from the definition that the random variables $X_0, \dots , X_n$ form a martingale.  Next, we bound the mean and variance in each step of the martingale.  WLOG we may assume $\pi^* = \id$.  We will use this simplification throughout the remainder of this section.

\subsubsection{Bound on Mean}
\begin{claim}\label{claim:martingale-mean}
Assume $\norm{v_{\id} - v_{\sigma^*}} \leq 1/(2(1 - \phi))$.  We have
\[
 -10(1 - \phi) \norm{v_{\id} - v_{\sigma^*}}^2 \leq \E_{\pi \sim M(\phi, \id)}[\Delta_{\id, \sigma^*}(\pi) ] \leq -0.01 ( 1 - \phi) \norm{v_{\id} - v_{\sigma^*}}^2
\]
\end{claim}
\begin{proof}
Note that by linearity of expectation,
\[
\E_{\pi \sim M(\phi, \id)}[\Delta_{\id, \sigma^*}(\pi) ] = \sum_{(x,y) \in S(\id, \sigma^*)} \left(\Pr_{\pi \sim M(\phi, \id)}(\pi^{-1}(x) > \pi^{-1}(y) ) -  \Pr_{\pi \sim M(\phi, \id)}(\pi^{-1}(x) < \pi^{-1}(y) ) \right) 
\]
Note that by definition, all pairs $(x,y) \in S(\id, \sigma^*)$ satisfy that $x < y$.  Also, by assumption, for all $(x,y) \in S(\id, \sigma^*)$, 
\[
y - x \leq 1/(1 - \phi) \,.
\]
By Claim \ref{claim:inversion-prob}, we have
\[
\E_{\pi \sim M(\phi, \id)}[\Delta_{\id, \sigma^*}(\pi) ] \leq -0.02 (1- \phi)\sum_{(x,y) \in S(\id, \sigma^*)} \left(y - x\right) = -0.01 ( 1 - \phi) \norm{v_{\id} - v_{\sigma^*}}^2
\]
and 
\[
\E_{\pi \sim M(\phi, \id)}[\Delta_{\id, \sigma^*}(\pi) ] \geq  -20(1 - \phi) \sum_{(x,y) \in S(\id, \sigma^*)} \left( y - x \right) = -10(1 - \phi) \norm{v_{\id} - v_{\sigma^*}}^2
\]
where the last step in both of the above follows from Claim \ref{claim:sum-identity}.
\end{proof}

\subsubsection{Bound on Variance}
Now we will bound the variance of the Martingale $X_t$.  In particular, in the following lemma, we show that the sum over all timesteps of the maximum change $|X_t - X_{t-1}|^2$ is bounded.  We will then be able to plug this bound into Azuma's inequality for martingale concentration.

\begin{lemma}\label{lem:martingale-variance}
Assume that $\id$ and $\sigma^*$ match on their first $\lceil 1/(1 - \phi) \rceil$ elements.  Also assume that 
\[
\norm{v_{\id} - v_{\sigma^*}} \leq \frac{1}{2 (1 - \phi)} \,.
\]
Then we have
\[
\sum_{t = 1}^n \max |X_t - X_{t - 1}|^2 \leq C \norm{v_{\id} - v_{\sigma^*}}^2 
\]
where $C$ is some universal constant.
\end{lemma}

First, we prove a few preliminary identities.

\begin{definition}\label{def:pairwise-inversion-counter}
For a permutation $\pi$ on $[n]$ and element $x \in [n]$ define $Q_{\pi}(x)$ to be the number of pairs of elements $(y,z)$ such that $y < x < z$ and $z$ occurs before $y$ in the permutation $\pi$.
\end{definition}

\begin{claim}\label{claim:crossing-bound}
For all permutations $\pi$ on $[n]$,
\[
\sum_{t = 1}^n Q_{\pi}(t)  \leq \frac{1}{2} \norm{v_{\id} - v_{\pi}}^2 
\]
\end{claim}
\begin{proof}
Note that we can rewrite the sum
\[
\sum_{(x,y) \in S(\id, \pi)} ( y - x) = \sum_{(x,y) \in S(\id, \pi)} \sum_{t = 1}^n 1_{x \leq t < y} = \sum_{t = 1}^n |(x,y) \in S(\id, \pi), x \leq t< y |  \geq \sum_{t = 1}^n  Q_{\pi}(t)
\]
so by Claim \ref{claim:sum-identity}, we have
\[
\sum_{t = 1}^n  Q_{\pi}(t) \leq \sum_{(x,y) \in S(\id, \pi)} ( y - x) = \frac{1}{2} \norm{v_{\id} - v_{\pi}}^2 
\]
as desired.
\end{proof}


\begin{claim}\label{claim:martingale-step-bound}
Assume that $ \id$ and $\sigma^*$ match on their first $\lceil 1/(1 - \phi) \rceil$ elements and that
\[
\norm{v_{\id} - v_{\sigma^*}} \leq \frac{1}{ 2(1 - \phi)} \,.
\]
Then we always have, for all $t$,
\[
|X_t - X_{t - 1}| \leq 2(v_{\sigma^*}^{\leftarrow}[t] + v_{\sigma^*}^{\rightarrow}[t]) + 2C(1 - \phi)Q_{\sigma^*}(t)
\]
where $C$ is some universal constant.
\end{claim}
\begin{proof}
First note that if $t < 1/(1 - \phi)$ then clearly $ |X_t - X_{t - 1}| = 0$.  Thus, from now on we may assume $t \geq 1/(1 - \phi)$.  
\\\\
We fix $\pi_{\id: t - 1}$ and then consider inserting element $t$ into the $t$ possible locations.  First note that
\[
\Delta_{\id , \sigma^*}(\pi) = \sum_{(x,y) \in S(\id, \sigma^*)} 1_{\pi^{-1}(x) > \pi^{-1}(y)} - 1_{\pi^{-1}(y) > \pi^{-1}(x)} =  \sum_{(x,y) \in S(\id, \sigma^*)} \left(2 \cdot 1_{\pi^{-1}(x) > \pi^{-1}(y)} - 1\right) \,.
\]
Thus,
\[
X_t - X_{t - 1} =  2\sum _{(x,y) \in S(\id, \sigma^*)} \E[1_{\pi^{-1}(x) > \pi^{-1}(y)} | \pi_{\id: t } ] - \E[1_{\pi^{-1}(x) > \pi^{-1}(y)} |  \pi_{\id: t - 1}] \,.
\]
Next, observe that the summand is nonzero only in the following cases 
\begin{itemize}
    \item $x= t$
    \item $y = t$
    \item $x < t < y$
\end{itemize}
The first two cases will be easy to handle since
\begin{equation}\label{eq:bound1}
\left \lvert \sum _{(t,y) \in S(\id, \sigma^*)} \E[1_{\pi^{-1}(t) > \pi^{-1}(y)} | \pi_{\id: t } ] - \E[1_{\pi^{-1}(t) > \pi^{-1}(y)} |  \pi_{\id: t - 1}] \right \rvert \leq  v_{\sigma^*}^{\rightarrow}[t] 
\end{equation}
and
\begin{equation}\label{eq:bound2}
\left \lvert \sum _{(x,t) \in S(\id, \sigma^*)} \E[1_{\pi^{-1}(x) > \pi^{-1}(t)} | \pi_{\id: t } ] - \E[1_{\pi^{-1}(x) > \pi^{-1}(t)} |  \pi_{\id: t - 1}]  \right \rvert \leq  v_{\sigma^*}^{\leftarrow}[t] \,. 
\end{equation}
Now it remains to handle the third case.  First, note that since 
\[
\norm{v_{\id} - v_{\sigma^*}} \leq \frac{1}{2(1 - \phi)} \,,
\]
it must be the case that $y \leq t + 1/(1 - \phi)$ for all $(x,y) \in S(\id, \sigma^*)$ with $x < t < y$.  For such a pair $(x,y)$, we now bound the difference
\[
\left \lvert \E[1_{\pi^{-1}(x) > \pi^{-1}(y)} | \pi_{\id: t } ] - \E[1_{\pi^{-1}(x) > \pi^{-1}(y)} |  \pi_{\id: t - 1}] \right \rvert \,.
\]
To upper bound the above, it suffices to upper bound the difference
\[
\left \lvert \E_{\pi \sim M(\phi, \id)}[1_{\pi^{-1}(x) > \pi^{-1}(y)} | \pi_{\id: t } ] - \E_{\pi' \sim M(\phi, \id)}[1_{\pi^{\prime -1}(x) > \pi^{\prime -1}(y)} |  \pi'_{\id: t}] \right \rvert
\]
where $\pi_{\id: t }$ and $\pi'_{\id: t}$ induce the same permutation on $\{1,2, \dots , t-1 \}$, differing only in the location of element $t$.  However this means that in $\pi_{\id:t}$ and $\pi'_{\id:t}$, the locations of element $x$ differ by at most $1$.  Thus, we can apply Claim \ref{claim:small-change} to deduce that the above quantity is at most $C(1 - \phi)$ for some universal constant $C$. \\\\
Thus, we conclude for all $(x,y) \in S$ with $x < t < y$,
\[
\left \lvert \E[1_{\pi^{-1}(x) > \pi^{-1}(y)} | \pi_{\id: t } ] - \E[1_{\pi^{-1}(x) > \pi^{-1}(y)} |  \pi_{\id: t - 1}] \right \rvert \leq C(1 - \phi) \,.
\]
Since there are exactly $Q_{\sigma^*}(t)$ pairs $(x,y) \in S(\id, \sigma^*)$ with $x < t < y$, we have
\[
\left \lvert \sum_{\substack{(x,y) \in S(\id, \sigma^*) \\ x < t < y}}  \E[]1_{\pi^{-1}(x) > \pi^{-1}(y)} | \pi_{\id: t } ] - \E[1_{\pi^{-1}(x) > \pi^{-1}(y)} |  \pi_{\id: t - 1}] \right \rvert \leq C(1 - \phi)Q_{\sigma^*}(t) \,.
\]
Combining the above with (\ref{eq:bound1}, \ref{eq:bound2}), we conclude 
\[
|X_t - X_{t - 1}| \leq  2(v_{\sigma^*}^{\leftarrow}[t] + v_{\sigma^*}^{\rightarrow}[t]) + 2C(1 - \phi)Q_{\sigma^*}(t)
\]
which completes the proof.
\end{proof}

We can now prove Lemma \ref{lem:martingale-variance}.
\begin{proof}[Proof of Lemma \ref{lem:martingale-variance}]
By Claim \ref{claim:martingale-step-bound}, we have
\begin{align*}
\sum_{t = 1}^n \max |X_t - X_{t - 1}|^2 &\leq  \sum_{t = 1}^n \left( 8(v_{\sigma^*}^{\leftarrow}[t] + v_{\sigma^*}^{\rightarrow}[t])^2 + 8C^2(1 - \phi)^2Q_{\sigma^*}(t)^2 \right)  \\ &\leq 8 C^2 (1-\phi)^2 \left( \sum_{t = 1}^n Q_{\sigma^*}(t) \right)^2 + 8\sum_{i = 1}^n (v_{\sigma^*}^{\leftarrow}[t] + v_{\sigma^*}^{\rightarrow}[t])^2 \,.
\end{align*}
Applying Claim \ref{claim:inversion-bound} and Claim \ref{claim:crossing-bound} to the above, we get
\begin{align*}
\sum_{t = 1}^n \max |X_t - X_{t - 1}|^2 \leq 2 C^2 (1 - \phi)^2 \norm{v_{\id} - v_{\sigma^*}}^4 + 32\norm{v_{\id} - v_{\sigma^*}}^2 \leq (32 + 2C^2)\norm{v_{\id} - v_{\sigma^*}}^2
\end{align*}
which completes the proof.
\end{proof}

\subsection{Martingale Concentration}

We now use Claim \ref{claim:martingale-mean} and Lemma \ref{lem:martingale-variance} in conjunction with Azuma's inequality to obtain concentration for
\[
X_n = \Delta_{\id, \sigma^*} (\pi) 
\]
which will allow us to prove Theorem \ref{thm:TV-lowerbound}.  

\begin{theorem}[Azuma's Inequality]\label{thm:Azuma}
For random variables $X_0, X_1, \dots , X_n$ that form a martingale, assume that 
\[
|X_t - X_{t - 1}| \leq c_t
\]
for all $t$.  Then 
\[
\Pr[ |X_n - X_0| \geq \eps ]\leq 2e^{\frac{-\eps^2 }{2(c_1^2 + \dots + c_n^2)}} \,.
\]
\end{theorem}
Combining Azuma's inequality with Lemma \ref{lem:martingale-variance} and Claim \ref{claim:martingale-mean}, we have the following concentration inequality.
\begin{corollary}\label{coro:tailbound}
Assume that the permutations $\id$ and $\sigma^*$ satisfy that 
\[
\norm{v_{\id} - v_{\sigma^*}} \leq \frac{1}{2(1 - \phi)} \,.
\]
There is a universal constant $C$ such that for any constant $c > 1$,
\[
\Pr_{\pi \sim M(\phi, \id)}[ |\Delta_{\id, \sigma^*}(\pi)| \geq c  \cdot C \norm{v_{\id} - v_{\sigma^*}} ] \leq 2e^{-c^2} \,. 
\]
\end{corollary}
\begin{proof}
Note that we can pad $1/(1-\phi)$ additional elements at the front of both $\id$ and $\sigma^*$ without changing anything so WLOG we can assume that $\id$ and $\sigma^*$ match on their first $1/(1-\phi)$ elements.  Recall that in the definition of the martingale $X_0,X_1, \dots, X_n$, we have
\[
X_n = \Delta_{\id, \sigma^*}(\pi), X_0 = \E_{\pi \sim M(\phi, \id)}[  \Delta_{\id, \sigma^*}(\pi)] \,.
\]
By Claim \ref{claim:martingale-mean}, we have
\[
|X_0| \leq 10(1-\phi) \norm{v_{\id} - v_{\sigma^*}}^2 \leq 5 \norm{v_{\id} - v_{\sigma^*}} \,.
\]
Combining Theorem \ref{thm:Azuma} and Lemma \ref{lem:martingale-variance}, we have that for some universal constant $C$ and any constant $c > 0$, 
\[
\Pr_{\pi \sim M(\phi, \id)}\left[ |X_n - X_0| \geq  c \cdot C\norm{v_{\id} - v_{\sigma^*}}\right]\leq 2e^{- c^2 } \,.
\]
Combining the above two inequalities (and adjusting the constant $C$) we get the desired statement.
\end{proof}

We are almost ready to prove Theorem \ref{thm:TV-lowerbound} but we will need one more basic inequality.
\begin{claim}\label{claim:basic-exp-inequality}
For all real numbers $x$, we have
\[
x^3(e^{2x} - 1 ) \geq (e^x - 1 - x)^2 \,.
\]
\end{claim}
\begin{proof}
Note that $e^x - 1 - x \geq 0$.  First, note that $e^{-x} \geq 1 - x$ which rearranges into 
\[
e^x - 1 - x \leq x(e^x - 1) \,.
\]
Next, observe that $(1 - x^2) e^x \leq 1 + x + x^2$.  To see this, note that if $|x| \geq 1$ then the LHS is non-positive and the RHS is always positive.  Otherwise $e^x \leq 1 + x + x^2$ and the inequality follows trivially.  This rearranges into
\[
e^x - 1 - x \leq x^2(e^x + 1) \,.
\]
Multiplying the two ienqualities (note that all sides are positive), we conclude
\[
 (e^x - 1 - x)^2 \leq x^3(e^{2x} - 1 )
\]
as desired.
\end{proof}

Now we can prove Theorem \ref{thm:TV-lowerbound}.

\begin{proof}[Proof of Theorem \ref{thm:TV-lowerbound}]
WLOG, we may assume $\pi^* = \id$.  If $\phi < 1/2$ then we must have $\id = \sigma^* $ and clearly the TV distance is $0$.  Thus, we may now assume $\phi \geq 1/2$.  This means that $-\log \phi \leq 2(1 - \phi)$, which we will use repeatedly later on.  
\\\\
Now define the random variable 
\[
\theta(\pi) = \Delta_{\id, \sigma^*}(\pi) \cdot (-\log \phi)
\]
for $\pi \sim M(\phi, \id)$.  Claim \ref{claim:martingale-mean} implies that 
\[
\E_{\pi \sim M(\phi, \id)}[\theta(\pi)] \leq (0.01 \log \phi) (1 - \phi) \norm{v_{\id} - v_{\sigma^*}}^2 \,.
\]
Also note that by (\ref{eq:ratio}),
\[
\E_{\pi \sim M(\phi, \id)}[e^{\theta(\pi)} - 1] = \sum_{\pi}\Pr_{M(\phi, \id)}[\pi] \left(  \frac{\Pr_{M(\phi, \sigma^*)}[\pi]}{\Pr_{M(\phi, \id)}[\pi]} - 1 \right) = \sum_{\pi} \left( \Pr_{M(\phi, \sigma^*)}[\pi] - \Pr_{M(\phi, \id)}[\pi]\right) = 0 \,.
\]
Thus,
\[
\E_{\pi \sim M(\phi, \id)}[e^{\theta(\pi)} - \theta(\pi) - 1] \geq (-0.01 \log \phi) (1 - \phi) \norm{v_{\id} - v_{\sigma^*}}^2 \geq 0.01 (1 - \phi)^2 \norm{v_{\id} - v_{\sigma^*}}^2\,.
\]
Now by Corollary \ref{coro:tailbound}, we have
\[
\E_{\pi \sim M(\phi, \id)}[ |\theta(\pi)|^3 (e^{\theta(\pi)} + 1 )  ] \leq C (1 - \phi)^3 \norm{v_{\id} - v_{\sigma^*}}^3 
\]
for some universal constant $C$.  This is because the variable $\theta(\pi)$ is sub-Gaussian with parameter 
\[
\sigma = O((-\log \phi) \norm{v_{\id} - v_{\sigma^*}} ) =  O((1 - \phi) \norm{v_{\id} - v_{\sigma^*}} ) \leq O(1) \,.
\]
Now by Claim \ref{claim:basic-exp-inequality} and Cauchy Schwarz, we have
\[
\E_{\pi \sim M(\phi, \id)}[ |e^{\theta(\pi)}  - 1 |] \geq \frac{\left( \E_{\pi \sim M(\phi, \id)}[e^{\theta(\pi)} - \theta(\pi) - 1]\right)^2}{\E_{\pi \sim M(\phi, \id)}[ |\theta(\pi)|^3 (e^{\theta(\pi)} + 1 )  ]} \geq (1 - \phi) \norm{v_{\id} - v_{\sigma^*}} / C
\]
for some universal constant $C$ (possibly different from the previous one).  However, again by (\ref{eq:ratio}) the LHS is 
\begin{align*}
\E_{\pi \sim M(\phi, \id)}[ |e^{\theta(\pi)}  - 1 |] = \sum_{\pi}\Pr_{M(\phi, \id)}[\pi] \cdot  \left \vert \frac{\Pr_{M(\phi, \sigma^*)}[\pi]}{\Pr_{M(\phi, \id)}[\pi]} - 1 \right \rvert = \sum_{\pi} \left \lvert\Pr_{M(\phi, \id)}[\pi] - \Pr_{M(\phi, \sigma^*)}[\pi] \right \rvert \\ = 2d_{\TV}( M(\phi, \id), M(\phi, \sigma^*)) 
\end{align*}
and we are done.

\end{proof}

\section{Omitted Proofs from Section \ref{sec:tail-bounds}}\label{appendix:tail-bounds}
In this section, we prove Lemma \ref{lem:Gaussian-tails}.  WLOG, we will assume $\pi^* = \id$ throughout the proof.  Since the entries of $v_{\pi}^{\leftarrow}$ (recall the shorthand in Definition \ref{def:abbreviate}) are independent, it suffices to bound the moment generating function of each entry of $v_{\pi}^{\leftarrow}$ to obtain the desired bound.  This is done in the following claim.

\begin{claim}\label{claim:exp-bound}
Let $M(\phi, \id)$ be a Mallows model on $n$ elements and assume $\phi \geq 0.5$.    For any $i \in [n]$ and parameter $-0.2 < \alpha < 0.2$, we have
\[
\E_{\pi \sim M(\phi, \id)}\left[ (1 / \phi)^{\alpha \left(v_{\pi}^{\leftarrow}[i] - \E[ v_{\pi}^{\leftarrow}[i]]\right)} \right] \leq e^{10 \alpha^2} \,.
\]
\end{claim}

To prove Claim \ref{claim:exp-bound}, first we observe the following general facts.

\begin{claim}\label{claim:convex-func-property}
Let $f: \R \rightarrow \R$ be a convex function.  Let $x_1, \dots , x_n$ be real numbers with $x_1 \leq x_2 \leq \dots \leq x_n$.  Let $\delta_1, \dots , \delta_n$ be real numbers with $\delta_1 \leq \delta_2 \leq \dots \leq \delta_n$.  Finally, let $w_1, \dots , w_n \geq 0$ be non-negative real numbers and assume $w_1 \delta_1 + \dots + w_n \delta_n = 0$.  Then
\[
w_1f(x_1) + \dots + w_nf(x_n) \leq w_1 f(x_1 + \delta_1) + \dots  + w_nf(x_n + \delta_n) \,.
\]
\end{claim}
\begin{proof}
We prove the claim by induction on $n$.  First, we do the base case when $n = 2$.  Note that clearly we must have $\delta_1 \leq 0 \leq \delta_2$.  Note that by convexity, we have the following relation
\[
f(x_1) \leq \frac{x_2 + \delta_2 - x_1}{x_2 + \delta_2 - (x_1 + \delta_1)} f(x_1 + \delta_1) + \frac{(-\delta_1)}{x_2 + \delta_2 - (x_1 + \delta_1)} f(x_2  +\delta_2) \,.
\]
Substituting the above, it suffices to prove that
\[
w_2 f(x_2) \leq w_1 \frac{(-\delta_1)}{x_2 + \delta_2 - (x_1 + \delta_1)} f(x_1 + \delta_1) + \left(w_2 - \frac{-w_1\delta_1}{x_2 + \delta_2 - (x_1 + \delta_1)} \right) f(x_2 + \delta_2) \,.
\]
However, using the fact that $w_1\delta_1 + w_2\delta_2 = 0$, the RHS of the above rearranges as
\[
\frac{w_2\delta_2}{x_2 + \delta_2 - (x_1 + \delta_1)} f(x_1 + \delta_1) + \frac{w_2(x_2 - (x_1 + \delta_1))}{x_2 + \delta_2 - (x_1 + \delta_1)}f(x_2 + \delta_2) 
\]
and now it follows immediately from convexity that the above quantity is at least $w_2f(x_2)$, completing the base case of the induction.  Now we do the inductive step.  Note that clearly $\delta_1 \leq 0$.  Now consider the sequence $\delta_2 + (w_1/w_2)\delta_1 , \delta_3, \dots , \delta_n$.  Note that this sequence is nondecreasing and satisfies
\[
w_2\left(\delta_2 + (w_1/w_2)\delta_1 \right) + w_3\delta_3 + \dots + w_n \delta_n = w_1\delta_1 + \dots + w_n \delta_n = 0 \,.
\]
Thus, we may apply the inductive hypothesis and get
\[
w_2f(x_2) + \dots + w_nf(x_n) \leq w_2f\left( x_2 + \delta_2 + (w_1/w_2)\delta_1\right)  + w_3f(x_3 + \delta_3) + \dots + w_nf(x_n + \delta_n) \,.
\]
Now it remains to prove that
\[
w_1f(x_1) + w_2f\left( x_2 + \delta_2 + (w_1/w_2)\delta_1\right) \leq w_1f(x + \delta_1) + w_2f(x + \delta_2) \,.
\]
However, this inequality follows immediately from the inductive hypothesis for $n = 2$ (with setting $x_2 \leftarrow x_2 + \delta_2 + (w_1/w_2)\delta_1$ and $\delta_2 = -(w_1/w_2) \delta_1$).  This completes the induction step and we are done.
\end{proof}

\begin{claim}\label{claim:majorization}
For any convex function $f$, the quantity
\[
\E_{\pi \sim M(\phi, \id)}\left[ f \left( v_{\pi}^{\leftarrow}[i] - \E[ v_{\pi}^{\leftarrow}[i]]\right) \right]
\]
is non-decreasing in $i$.
\end{claim}
\begin{proof}
Recall Lemma \ref{lem:insertion-sampling} and Definition \ref{def:insertion-distribution}.  We know that the distribution of $v_{\pi}^{\leftarrow}[i]$ is exactly $\mcl{D}_i$.  Let $\mu_i$ be the mean of $\mcl{D}_i$.  Now we can write
\[
\E_{\pi \sim M(\phi, \id)}\left[ f \left( v_{\pi}^{\leftarrow}[i] - \E[ v_{\pi}^{\leftarrow}[i]]\right) \right] = \sum_{j = 0}^{i - 1 } \frac{\phi^{j}}{1 + \phi + \dots + \phi^{i - 1}} f(j - \mu_i) \,.
\]
We can also write
\begin{align*}
&\E_{\pi \sim M(\phi, \id)}\left[ f \left( v_{\pi}^{\leftarrow}[i + 1] - \E[ v_{\pi}^{\leftarrow}[i + 1]]\right) \right] = \sum_{j = 0}^{i  } \frac{\phi^{j}}{1 + \phi + \dots + \phi^{i }} f(j - \mu_{i + 1}) \\ &= \sum_{j = 0}^{i} \left(\frac{\phi^{j}}{1 + \phi + \dots + \phi^{i - 1}}  \left( \frac{1 + \phi + \dots + \phi^{i - 1 - j}}{1 + \phi + \dots + \phi^i}\right) +  \frac{\phi^{j - 1}}{1 + \phi + \dots + \phi^{i - 1}}  \left( \frac{\phi^{i+1 - j} + \dots + \phi^{i}}{1 + \phi + \dots + \phi^i}\right) \right)f(j - \mu_{i + 1}) \\ & =\sum_{j = 0}^{i - 1} \frac{\phi^{ j}}{1 + \phi + \dots + \phi^{i - 1}} \left(  \frac{1 + \phi + \dots + \phi^{i - 1 - j}}{1 + \phi + \dots + \phi^i} f(j - \mu_{i + 1}) +  \frac{\phi^{i - j} + \dots + \phi^i}{1 + \phi + \dots + \phi^i}f(j + 1 - \mu_{i + 1})\right) \\ & \geq \sum_{j = 0}^{i-1} \frac{\phi^{j}}{1 + \phi + \dots + \phi^{i - 1}} f\left( j - \mu_{i + 1} + \frac{\phi^{i - j} + \dots + \phi^i}{1 + \phi + \dots + \phi^i} \right) \,.
\end{align*}
where the last step above follows from convexity of $f$.  Now we can apply Claim \ref{claim:convex-func-property} to conclude that 
\begin{equation}\label{eq:use-convexity}
\sum_{j = 0}^{i -1 } \frac{\phi^{j}}{1 + \phi + \dots + \phi^{i - 1}} f(j - \mu_i) \leq \sum_{j = 0}^{i - 1} \frac{\phi^{j}}{1 + \phi + \dots + \phi^{i - 1}} f\left( j - \mu_{i + 1} + \frac{\phi^{i - j} + \dots + \phi^i}{1 + \phi + \dots + \phi^i}\right) \,.
\end{equation}
To see why this application is valid, note that for all $j$ we are setting
\[
x_j =  j - \mu_i \quad , \quad w_j = \frac{\phi^{j}}{1 + \phi + \dots + \phi^{i - 1}}    \quad ,  \quad \delta_j = \frac{\phi^{i - j} + \dots + \phi^i}{1 + \phi + \dots + \phi^i} + \mu_i - \mu_{i + 1} \,.
\]
It is clear that all of these sequences are non-decreasing.  It remains to verify that $w_1\delta_1 + \dots + w_i\delta_i = 0$.  To see this, note that
\begin{align*}
w_1\delta_1 + \dots + w_i\delta_i &=   - (\mu_{i+ 1} - \mu_i) + \sum_{j = 0}^{i-1} \frac{\phi^{ j} (\phi^{i - j}  + \dots + \phi^i)}{(1 + \phi + \dots + \phi^{i - 1})(1 + \phi + \dots + \phi^i)} \\ &= - (\mu_{i+ 1} - \mu_i) + \frac{i\phi^i + (i-1)\phi^{i+1} + \dots + \phi^{2i - 1}}{(1 + \phi + \dots + \phi^{i - 1})(1 + \phi + \dots + \phi^i)} 
\end{align*}
but
\begin{align*}
\mu_{i+ 1} - \mu_i &= \frac{i \phi^i}{1 + \phi + \dots + \phi^i} - \sum_{j = 0}^{i-1} \left(\frac{\phi^{j}}{1 + \phi + \dots + \phi^{i - 1}} - \frac{\phi^{j}}{1 + \phi + \dots + \phi^i} \right) j \\ & = \frac{i\phi^i + (i-1)\phi^{i+1} + \dots + \phi^{2i - 1}}{(1 + \phi + \dots + \phi^{i - 1})(1 + \phi + \dots + \phi^i)}
\end{align*}
and combining this with the above implies $w_1\delta_1 + \dots + w_i\delta_i = 0$.  Thus, we indeed have proven (\ref{eq:use-convexity}) and from it we conclude
\[
\E_{\pi \sim M(\phi, \id)}\left[ f \left( v_{\pi}^{\leftarrow}[i + 1] - \E[ v_{\pi}^{\leftarrow}[i + 1]]\right) \right] \geq \E_{\pi \sim M(\phi, \id)}\left[ f \left( v_{\pi}^{\leftarrow}[i] - \E[ v_{\pi}^{\leftarrow}[i]]\right) \right]
\]
completing the proof of the claim.
\end{proof}

Now we can prove the tail bound in Claim \ref{claim:exp-bound}.

\begin{proof}[Proof of Claim \ref{claim:exp-bound}]
The function $f(x) = (1/\phi)^{\alpha x}$ is clearly convex.  Thus, by repeatedly applying Claim \ref{claim:majorization}, it suffices to consider the limit as $i, n \rightarrow \infty$ i.e. when the distribution of $v_{\pi}^{\leftarrow}[i]$ is exactly $\mcl{D}_{\infty}$.  Note that by Claim \ref{claim:explicit-means}, the mean of $\mcl{D}_{\infty}$ is $\phi/(1 - \phi)$.  Now the desired expectation is
\begin{equation}\label{eq:first-step}
\E_{\substack{\pi \sim M(\phi, \id) \\ i \rightarrow \infty}}\left[ (1 / \phi)^{\alpha \left(v_{\pi}^{\leftarrow}[i] - \E[ v_{\pi}^{\leftarrow}[i]]\right)} \right]  = \phi^{\alpha \phi/(1 - \phi)}\sum_{j = 0}^{\infty} (1 - \phi) \phi^{(1 - \alpha)j} 
\end{equation}
Note that since $0.5 \leq \phi \leq 1$, we have
\[
1 \leq \frac{-\log \phi}{1 - \phi} \leq 2 \,.
\]
First, if $\alpha \geq 0$ then
\begin{align*}
\sum_{j = 0}^{\infty} (1 - \phi) \phi^{(1 - \alpha)j}  &= \sum_{j = 0}^{\infty} (1 - \phi) \phi^j \left(\sum_{k = 0}^{\infty}\frac{ (-\alpha \log \phi)^k j^k}{k!} \right) \\ &\leq \left(1 + \sum_{k = 1}^{\infty} (1 - \phi)(-\alpha \log \phi)^k \sum_{j = 0}^{\infty}\binom{j + k - 1}{k}\phi^j\right)  \\ &= \left( 1 + \sum_{k = 1}^{\infty} (-\alpha \log \phi)^k \frac{\phi}{(1 - \phi)^{k}} \right) \\ & = 1 + \frac{-\alpha \phi \log \phi}{1 - \phi} + \frac{2(\alpha \log \phi)^2 \phi}{(1 - \phi)^2}  \\ & \leq 1 + \frac{-\alpha \phi \log \phi}{1 - \phi} + 8 \alpha^2 \,. 
\end{align*}
where the last two steps use the fact that $(- \log \phi)/(1 - \phi) \leq 2$ and $\alpha \leq 0.2$.  We also have
\begin{align*}
\phi^{\alpha \phi/(1 - \phi)} \leq 1 - \frac{-\alpha \phi \log \phi}{1 - \phi} + \frac{(\alpha \phi \log \phi)^2}{2(1 - \phi)^2} \leq 1 - \frac{-\alpha \phi \log \phi}{1 - \phi} +  2 \alpha^2 \,.
\end{align*}
Thus, combining the previous two inequalities with (\ref{eq:first-step}) and using AM-GM, we conclude
\[
\E_{\substack{\pi \sim M(\phi, \id) \\ i \rightarrow \infty}}\left[ (1 / \phi)^{\alpha \left(v_{\pi}^{\leftarrow}[i] - \E[ v_{\pi}^{\leftarrow}[i]]\right)} \right] \leq (1 + 5\alpha^2)^2 \leq e^{10 \alpha^2} \,.
\]
Now it remains to consider the case where $\alpha \leq 0$.  In this case, using similar computations, we have
\begin{align*}
\sum_{j = 0}^{\infty} (1 - \phi) \phi^{(1 - \alpha)j}  &\leq  \sum_{j = 0}^{\infty} (1 - \phi) \phi^j \left(\sum_{k = 0}^{2}\frac{ (-\alpha \log \phi)^k j^k}{k!} \right) \\ &\leq \left(1 + \sum_{k = 1}^{2} (1 - \phi)(-\alpha \log \phi)^k \sum_{j = 0}^{\infty}\binom{j + k - 1}{k}\phi^j\right)  \\ &= \left( 1 + \sum_{k = 1}^{2} (-\alpha \log \phi)^k \frac{\phi}{(1 - \phi)^{k}} \right) \\& \leq 1 - \frac{\alpha \phi \log \phi}{1 - \phi} + 4 \alpha^2 
\end{align*}
and 
\[
\phi^{\alpha \phi/(1 - \phi)} \leq 1 + \frac{\alpha \phi \log \phi}{1 - \phi} + \frac{(\alpha \phi \log \phi)^2}{(1 - \phi)^2} \leq 1 + \frac{\alpha \phi \log \phi}{1 - \phi} + 4 \alpha^2 \,.
\]
Combining the previous two inequalities we deduce that in the case where $\alpha \leq 0$, we also have
\[
\E_{\substack{\pi \sim M(\phi, \id) \\ i \rightarrow \infty}}\left[ (1 / \phi)^{\alpha \left(v_{\pi}^{\leftarrow}[i] - \E[ v_{\pi}^{\leftarrow}[i]]\right)} \right] \leq (1 + 4\alpha^2)^2 \leq e^{8 \alpha^2} \,.
\]
This completes the proof.
\end{proof}

The proof of Lemma \ref{lem:Gaussian-tails} will now follow from a standard Chernoff bound.

\begin{proof}[Proof of Lemma \ref{lem:Gaussian-tails}]
Let $-0.2 \leq \alpha \leq 0.2$ be a parameter that will be set later.  Let the entries of $v$ be $v = (v[1], \dots , v[n])$.  Since the entries of $v_{\pi}^{\leftarrow}$ are independent
\[
\E_{\pi \sim M(\phi, \id)} \left[ (1/\phi)^{\alpha( v \cdot v_{\pi}^{\leftarrow} - \E[v \cdot v_{\pi}^{\leftarrow}]) }\right] = \prod_{i = 1}^n \E\left[ (1 / \phi)^{\alpha v[i] \left(v_{\pi}^{\leftarrow}[i] - \E[ v_{\pi}^{\leftarrow}[i]]\right)} \right]   \leq \prod_{i = 1}^n e^{10 \alpha^2 v[i]^2} = e^{10 \alpha^2} 
\]
where in the above we used Claim \ref{claim:exp-bound} and the fact that $v$ is a unit vector.  Thus, 
\[
\Pr\left[v \cdot v_{\pi}^{\leftarrow} - \E[v \cdot v_{\pi}^{\leftarrow}] \geq  \frac{t}{1 - \phi}\right] \leq \frac{e^{10 \alpha^2}}{(1/\phi)^{\alpha t/(1 - \phi)}}  \leq e^{10 \alpha^2 - \alpha t} \,.
\]
Setting $\alpha = 0.2$ in the above, we conclude
\[
\Pr\left[v \cdot v_{\pi}^{\leftarrow} - \E[v \cdot v_{\pi}^{\leftarrow}] \geq  \frac{t}{1 - \phi}\right] \leq 2 e^{-0.2t}
\]
and similarly, we get
\[
\Pr\left[v \cdot v_{\pi}^{\leftarrow} - \E[v \cdot v_{\pi}^{\leftarrow}] \leq  \frac{-t}{1 - \phi}\right] \leq 2 e^{-0.2t} \,.
\]
Combining the previous two inequalities, we immediately get the desired statement for $v_{\pi}^{\leftarrow}$.  The analogous statement for $v_{\pi}^{\rightarrow}$ follows by symmetry (since we can imagine running the insertion sampling procedure from the back).
\end{proof}

\section{Omitted Proofs from Section \ref{sec:main-alg}}\label{appendix:main-alg}

\subsection{Generic Stability Bounds}

In this section, we prove Corollary \ref{coro:high-dim-stability}.  We start with the following concentration inequality. 
\begin{claim}\label{claim:prev-stability-bound}
Let $\mcl{D}$ be a distribution on $\R$ and $0 < c < 1$ be a positive constant such that for all real numbers $t > 1$,
\[
\Pr_{x \sim \mcl{D}} [ |x| \geq   t   ] \leq e^{-t^{c}} \,. 
\]
Let $0 < \delta, \eps < 1$ be parameters.  Given a set $S$ of 
\[
|S| \geq 10^{10/c}\log^{1/c}(1/\delta) (\log \log (1/\delta))^3/\eps^2 
\]
independent samples from $\mcl{D}$, with probability at least $1 -  \delta$, any subset $S' \subseteq S$ of size at least $(1-\eps)|S|$ satisfies
\[
\left \lvert \mu_{\mcl{D}} -  \frac{1}{|S'|} \sum_{x \in S'} x   \right \rvert \leq \eps \log^{1/c} (1/\eps) \cdot (10/c)^{O(1/c)}  \,.
\]
\end{claim}
\begin{proof}
First, with probability $1 - \delta/4$, all samples are within the interval 
\[
[-\log^{1/c} (1/(\eps\delta)^2), \log^{1/c} (1/(\eps \delta)^2 ) ] \,.
\]
Now for a parameter $\alpha$, we bound the number samples that have magnitude at least $\alpha \log^{1/c}(1/\eps)$.  The probability that this number of samples is more than $\frac{10\eps |S|}{\alpha \log^2 \alpha}$ is at most
\[
\binom{|S|}{\frac{10\eps |S|}{\alpha \log^2 \alpha}}\left( \eps^{\alpha^c}\right)^{\frac{10\eps |S|}{\alpha \log^2 \alpha}} \leq \left(\alpha^2  \eps^{\alpha^c - 1} \right)^{\frac{10\eps |S|}{\alpha \log^2 \alpha}} \leq \frac{\delta}{2\alpha}  
\]
where the last step holds because it suffices to consider when $\alpha \leq 2 \log^{1/c}(1/\delta)$.  Let $C = (10/c)^{10/c}$.  We can union bound the above over $\alpha = C, 2C, \dots $ and then if this happens, then we get that for any subset $T \subset S$ with $|T| = \eps |S|$, then 
\begin{equation}\label{eq:subset-bounds}
\left\lvert \sum_{x \in T} x \right \rvert \leq   |T| \log^{1/c} (1/\eps) \cdot O(C) \,.
\end{equation}
To complete the proof, it suffices to upper bound the difference between the mean of $S$ and $\mu_{\mcl{D}}$.  For each sample $x$, let $x'$ be equal to $x$ if $|x| \leq C \log^{1/c} (1/\eps)$ and $0$ otherwise.  Let $\mu_{\mcl{D}}'$ be the mean of this new distribution.  By the tail bounds on $\mcl{D}$, we have that 
\[
|\mu_{\mcl{D}} - \mu_{\mcl{D}'} | \leq  \eps \log^{1/c} (1/\eps) \cdot O(C)
\]
and by a Chernoff bound, we have that with probability $1 - \delta/4$, at most $\eps|S|$ elements have $x \neq x'$ and
\[
\left \lvert \mu_{\mcl{D}'} -  \frac{1}{|S|} \sum_{x \in S} x'   \right \rvert \leq \eps \log^{1/c} (1/\eps) \cdot O(C) \,.
\]
Combining the above with \eqref{eq:subset-bounds} gives the desired inequality.
\end{proof}

As an immediate consequence of the above we have

\begin{corollary}\label{coro:sub-exponential-stability}
Let $\mcl{D}$ be a distribution on $\R$ with mean $\mu$ and variance $\sigma^2$ and assume that it has the property that for all $t > 1$,
\[
\Pr_{x \sim \mcl{D}}[ |x - \mu| \geq t ] \leq e^{-t} \,. 
\]
Consider drawing a set of samples $S$ i.i.d from $\mcl{D}$ with  $|S| \geq (n/\eps)^{2}\poly(\log (n/\eps))$.  Then with probability $1 - 2^{-20n \log n/\eps }$, any subset $S' \subset S$ with $|S'| \geq (1-\eps)|S|$ satisfies
\begin{align*}
\left \lvert \mu - \frac{1}{|S|'} \sum_{x \in S'} x \right \rvert &\leq C \eps \log 1/\eps \\
\left \lvert \sigma^2 - \frac{1}{|S|'} \sum_{x \in S'} (x - \mu)^2 \right \rvert &\leq C \eps \log^2 1/\eps \,.
\end{align*}
where $C$ is some sufficiently large universal constant.
\end{corollary}
\begin{proof}
The first statement is an immediate consequence of Claim \ref{claim:prev-stability-bound}.  The second statement follows by considering the variable $y = (x - \mu)^2$ which has tail decay of the form $e^{-\sqrt{y}}$ i.e. we can apply Claim \ref{claim:prev-stability-bound} with $c = 1/2$.
\end{proof}

Now we can prove Corollary \ref{coro:high-dim-stability}.
\begin{proof}[Proof of Corollary \ref{coro:high-dim-stability}]
Consider a $\delta$-net, say $T$, of the unit ball where $\delta = 0.1$.  Note that we can ensure $|T| \leq (3/\delta)^n$.  Applying Corollary \ref{coro:sub-exponential-stability}  and union bounding over all vectors in $T$, we get that with $1 - 2^{-10n \log n/\eps}$ probability, the two desired conditions hold for all vectors $v \in T$.  Now it remains to prove that they actually hold for all unit vectors.  Fix the set $S'$.  Define $\Delta \in \R^n$ and $\Gamma \in \R^{n \times n}$ as 
\begin{align*}
&\Delta = \mu - \frac{1}{|S|'} \sum_{x \in S'} x\\
&\Gamma = \Sigma - \frac{1}{|S|'} \sum_{x \in S'} (x- \mu) (x - \mu)^T \,.
\end{align*}
Since $T$ is a $\delta$-net, there is some $v \in T$ such that $v \cdot \Delta \geq \frac{1}{2} \norm{\Delta}$ and thus $\norm{\Delta} \leq O(\eps \log 1/\eps)$.  This immediately gives the first condition.  Now we verify the second.  Note that $\Gamma$ is a symmetric matrix.  Let $v$ be a unit vector corresponding to its largest eigenvalue i.e. such that $v^T\Gamma v = \norm{\Gamma}_{\op}$.  Let $v'$ be the vector in $T$ closest to $v$.  We have
\[
v^{\prime T} \Gamma v' = v^T \Gamma v + (v' - v)^T\Gamma v + v^{\prime T} \Gamma (v' - v) \geq v^T \Gamma v - 2\norm{v' - v}\norm{\Gamma}_{\op} \geq (1 - 2\delta)\norm{\Gamma}_{\op}  
\]
and thus $\norm{\Gamma}_{\op} \leq O(\eps \log^2 1/\eps)$ which verifies the second condition.
\end{proof}

\end{document}